\documentclass{amsart}

\usepackage{amssymb}%
\usepackage[mathscr]{eucal}
\usepackage{graphicx}
\usepackage{mathrsfs}
\usepackage{psfrag}
\usepackage{fullpage}

\theoremstyle{plain}
\newtheorem{theorem}{Theorem}[section]
\newtheorem{definition}{Definition}[section]
\newtheorem{proposition}[theorem]{Proposition}
\newtheorem{lemma}[theorem]{Lemma}
\newtheorem{assumption}[theorem]{Assumption}
\newtheorem{example}[theorem]{Example}
\newtheoremstyle{remm}%
{\topsep}
{\topsep}
{\sffamily}
{}
{\bfseries}
{}
{ }
{}
\theoremstyle{remm}
\newtheorem{rem}[theorem]{Remark}
\newenvironment{remark}
{\begin{rem}\setlength{\hangindent}{30 pt}}
{\end{rem}}

\newcommand{\lb}{\left\{}
\newcommand{\rb}{\right\}}
\newcommand{\Def}{\overset{\text{def}}{=}}

\newcommand{\R}{\mathbb{R}}
\newcommand{\Z}{\mathbb{Z}}
\newcommand{\N}{\mathbb{N}}

\newcommand{\KK}{K}
\newcommand{\Err}{\mathcal{E}}
\newcommand{\err}{\textsc{\tiny E}}

\newcommand{\vrho}{\varrho}
\newcommand{\vkap}{\varkappa}
\newcommand{\eps}{\varepsilon}

\newcommand{\Borel}{\mathscr{B}}
\newcommand{\Pspace}{\mathscr{P}}
\newcommand{\BP}{\mathbb{P}}
\newcommand{\BE}{\mathbb{E}}
\newcommand{\filt}{\mathscr{F}}
\newcommand{\gilt}{\mathscr{G}}

\newcommand{\Xsp}{\mathsf{X}}
\newcommand{\calI}{\mathcal{I}}

\newcommand{\Prem}{\textbf{P}^{\text{prem}}}
\newcommand{\Prot}{\textbf{P}^{\text{prot}}}
\newcommand{\rate}{\textsf{R}}
\newcommand{\PTimes}{\mathcal{T}}
\newcommand{\ff}{\mathfrak{f}}
\newcommand{\tL}{\bar L}

\newcommand{\idio}{\text{I}}
\newcommand{\system}{\text{S}}

\newcommand{\CS}{\mathcal{S}}
\newcommand{\Zn}{Z^{(n)}}
\newcommand{\fI}{\mathfrak{I}}
\newcommand{\granup}{\lceil N\alpha \rceil}
\newcommand{\grandn}{\lfloor N\alpha \rfloor}
\newcommand{\CA}{\mathcal{A}}
\newcommand{\CM}{\mathcal{M}}

\allowdisplaybreaks
\begin{document}
\title{Exact Pricing Asymptotics of Investment-Grade Tranches of Synthetic
CDO's Part I: A Large Homogeneous Pool}

\author{Richard B. Sowers}
\address{Department of Mathematics\\
    University of Illinois at Urbana--Champaign\\
    Urbana, IL 61801}
\email{r-sowers@illinois.edu}
\date{\today. Submitted.}

\thanks{The author would thank to Professor Tomasz Bielecki of the Illinois
Institute of Technology for inviting him to attend the Conference on Credit Risk
held at the University of Chicago in the Fall of 2007.  The author would also
like to thank Professor Neil Pearson of the University of Illinois at Urbana-Champaign for his time in explaining CDO's.}

\begin{abstract} We use the theory of large deviations to study the pricing
of investment-grade tranches of synthetic CDO's.  In this paper, we consider a
simplified model which will allow us to introduce some of the concepts and calculations. 
\end{abstract}

\maketitle

\section{Introduction}
It has been difficult to read the recent financial news without finding mention of Collateralized Debt Obligations (CDO's).  These financial instruments
provide ways of aggregating risk from a large number of sources and reselling
it in a number of parts, each part having different risk-reward characteristics.
Notwithstanding the role of CDO's in the recent market meltdown, the near
future will no doubt see the financial
engineering community continuing to develop structured investment
vehicles like CDO's.  Unfortunately, computational challenges in this area are formidable.
The main types of these assets have several common problematic features:
\begin{itemize}
\item they pool a large number of assets
\item they tranche the losses.
\end{itemize}
The ``problematic'' nature of this combination is that the trancheing procedure is nonlinear;
and as is usual, the effect of a nonlinear transformation on a high-dimensional
system is often difficult to understand.  Ideally, one would like a theory which gives,
if not explicit answers, at least some guidance.  Lacking theory, one is
often forced to search for models which are computationally
feasible, structurally robust, and which can be reasonably well-fitted to data.

We here consider a \emph{large deviations} (cf. \cite{MR1739680, MR1619036, MR758258}) analysis of certain aspects of synthetic CDO's.
The theory of large deviations is a collection of ideas which are
often useful in studying rare events.  The rare events of interest here involve losses in
(and hence pricing of)
investment-grade (senior or super-senior) tranches of synthetic CDO's.
We would like to see how far we can take a rigorous analysis when we use
mathematical tools, viz., large deviations, which are designed expressly to study rare events.
The theory of large deviations usually gives a very refined
analysis of rare events (more refined, for example, than one based on mean-variance calculations); what does this analysis look like for CDO's?

In the course of our analysis, we will see that
large deviations theory provides a natural framework for studying large amounts
of idiosyncratic randomness.  Moreover, the theory of large deviations
provides a way to compare rare events and see how they transform.
We believe this to be an important component of a larger
analysis of CDO's, particularly in cases where correlation comes from only a few sources (we will pursue a simple form of this idea in Subsection \ref{S:Correlated}).  In a sequel to this paper we will consider the more challenging case of a heterogeneous pool of assets.

This is not the first attempt to apply large deviations to structured finance.
Losses in pools of large assets like CDO's have been considered in \cite{MR2022976}, \cite{Glasserman}\footnote{Glasserman in \cite{Glasserman} makes important headway in understanding correlation.}, and \cite{MR2384674} (see also \cite{MR1644496} for another application of large deviations to finance).  Moreover, effects of tranching have been considered in
\cite{Vielex} and \cite{YangHurdZhang}, both of which discuss saddlepoint
effects of tranching once the distribution of the loss process is known.
Our interest is to identify, as much as possible, exact asymptotic formulae
for the price of the CDO by focussing on the effects of large amounts
of idiosyncratic randomness.  We find that if we interpret
the loss process as an occupation measure, \emph{Sanov's theorem}
suggests how to proceed.
Furthermore, it allows us to develop something of a bottom-up analysis
which directly connects the CDO price to the default probababilities of the underlying bonds.  It also naturally leads to a number of calculations which reflect the
\emph{dynamics} of the default probabilities (as opposed
to a snapshot of the default probabilities at expiry).

Finally, the \emph{ab initio} nature of our calculations bears note\footnote{See in particular Remark \ref{R:Comments} and the comments at the beginning of Section \ref{S:HAS}.}.
A number of models, such as the generalized Poission loss model \cite{Brigo}, the Hawkes process \cite{Giesecke} and others
(cf. \cite{Caflisch, FilOverSch}), which successfully capture some of the complexity of CDO's have been developed and implemented.
Our approach is limited to investment-grade tranches, and hopefully will complement some of these models and contribute to their study.

\section{CDS to CDO---a Review}
A standard review of credit default swaps and synthetic CDO's will help us fix notation, which comes from \cite{Brigo}.
Let's fix underlying probability triple $(\Omega,\filt,\BP)$, where $\BP$ represents the risk-neutral probability measure and $\BE$ is the associated expectation operator..

\subsection{Credit Default Swaps}\label{S:CDS}
A Credit Default Swap (CDS) is a contract between a \emph{protection seller}
and a \emph{protection buyer} based on the default of a reference bond
(a \emph{name}).  Under the contract, the protection seller pays the protection buyer $\$1$
(the \emph{notional}) when the bond defaults\footnote{We assume for simplicity no recovery.} (a nonnegative random time $\tau$), as long as this
default occurs before\footnote{We require default to be \emph{strictly} before expiry; that will save us some calculations resulting from potentially positive probability of default exactly at expiry.} the expiry of the contract (time $T$).  This is the
\emph{protection leg} of the contract.  In return, the protection buyer pays the
protection seller a premium $S$ at a finite collection $\PTimes$ of times (such that $t\le T$ for all $t\in \PTimes$ until
the default occurs.  This is the \emph{premium} leg of the contract; see Figure \ref{fig:CDS}.
\begin{figure}
\includegraphics[scale=0.8]{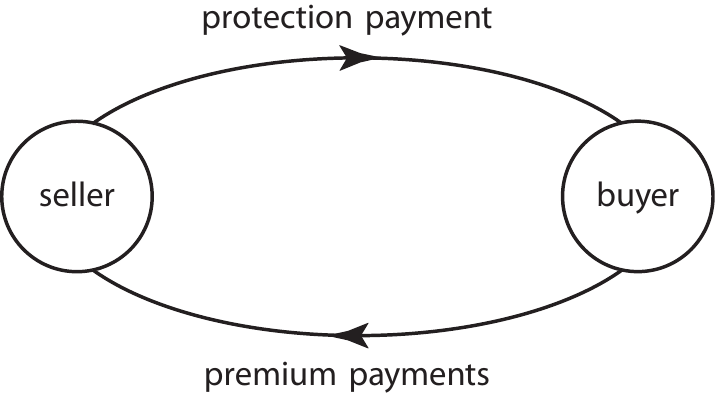}
\caption{Credit Default Swap}
\label{fig:CDS}
\end{figure}
To write this mathematically, define the loss process
\begin{equation*} L^\circ_t \Def \chi_{\{t\ge \tau\}} = \begin{cases} 1 &\text{if $t\ge \tau$} \\ 0 &\text{if $t<\tau$}\end{cases}\end{equation*}
for all $t\in \R$ (of course then $L^\circ_t=0$ for $t<0$).
The present value of the protection and premium legs are thus
\begin{gather*}e^{-\rate \tau}\chi_{\{\tau< T\}} = \int_{s\in [0,T)} e^{-\rate s}dL^\circ_s \\
S\sum_{t\in \PTimes} e^{-\rate t}\chi_{\{\tau>t\}}= S\sum_{t\in \PTimes}e^{-\rate t}\{1-\chi_{\{\tau\le t\}}\}= S\sum_{t\in \PTimes} e^{-\rate t}\left(1-L^\circ_t\right) \end{gather*}
where $\rate$ is the riskless interest rate\footnote{It is not difficult to see that
the maps $\omega\mapsto e^{-\rate \tau(\omega)}\chi_{\{\tau(\omega)\le T\}}$ and $\omega\mapsto \sum_{t\in \PTimes} e^{-\rate t}\chi_{\{\tau(\omega)>t\}}$ are measurable maps from $\Omega$ to $\R$; thus the expectations make sense.}.  The value of $S$ is defined by requiring that the expectation of these two legs agree (under the risk-neutral measure).

\subsection{Synthetic CDO's}

It is an easy step to modify this notation to construct a synthetic CDO.  Consider $N$ credit default swaps (each one on a different name).  Each CDS has notional value $1/N$, and the default of the $n$-th name occurs at a random nonnegative time $\tau_n$.  The \emph{notional} loss process is thus
\begin{equation*} L^{(N)}_t \Def \frac{1}{N}\sum_{n=1}^N \chi_{\{\tau_n\le t\}} \end{equation*}
for all $t\in \R$ (as in our above discussion of credit default swaps, $L^{(N)}_t=0$ for $t<0$).
Note that $0\le L^{(N)}_t\le 1$ for all $t\ge 0$.  Fix \emph{attachment} and \emph{detachment} points $\alpha$ and $\beta$ in $[0,1]$ such that $\alpha<\beta$.  We then define the \emph{tranched} loss process $\tL^{(N)}$ as
\begin{equation*} \tL^{(N)}_t \Def \frac{(L^{(N)}_t-\alpha)^+-(L^{(N)}_t-\beta)^+}{\beta-\alpha} = \begin{cases} 0 &\text{if $L^{(N)}_t<\alpha$} \\
\frac{L^{(N)}_t-\alpha}{\beta-\alpha} &\text{if $\alpha\le L^{(N)}_t\le \beta$} \\
1 &\text{if $L^{(N)}_t\ge \beta$}\end{cases} \end{equation*}
for all $t\in \R$.  The protection and premium legs of a synthetic CDO are basically given by replacing the loss process $L^\circ$ in a credit default swap with $\tL$.
Namely, define
\begin{equation*}\Prot_N \Def \int_{s\in [0,T)} e^{-\rate s}d\tL^{(N)}_s \qquad\text{and}\qquad
\Prem_N\Def \sum_{t\in \PTimes} e^{-\rate t}\left(1-\tL^{(N)}_t\right); \end{equation*}
$S_N\Prem_N$ is the present value of the premium leg (where $S_N$ 
are the premiums) and $\Prot_N$ is the present value of the protection leg.
The protection leg thus makes payments when defaults occur, as long as at
least $\alpha$ (in percent) of the names have already defaulted, and only
as long as no more than $\beta$ (in percent) of the names have defaulted.
These payments are proportioned so that they add up to at most $\$1$.  The
premium payments, on the other hand, are made only on the proportion
of names which are still insured (i.e., which have not yet defaulted).
The premium $S_N$ should then be given by equating the risk-neutral expectation of two legs; i.e.,
\begin{equation}\label{E:SN} S_N= \frac{\BE[\Prot_N]}{\BE[\Prem_N]}. \end{equation}
Note that $L^{(N)}_t$ is measurable for each $t\in \R$.  Since $\tL^{(N)}_t$ is a continuous transformation of $L^{(N)}_t$, it is also measurable. Since $0\le e^{-\rate s}\le 1$, $0\le \tL\le 1$, and $\tL$ is nondecreasing,  $\Prot_N$ and $\Prem_N$
both take values in $[0,1]$.  Moreover, the measurability of $\tL$ implies that $\Prot_N$ and $\Prem_N$ are measurable.  Thus both $\BE[\Prot_N]$ and $\BE[\Prem_N]$ are well-defined, finite, and nonnegative.
Our goal is to evaluate $S_N$ when $N$ is large.  This will be accomplished in \eqref{E:Sas}.
\begin{figure}
\includegraphics[scale=0.7]{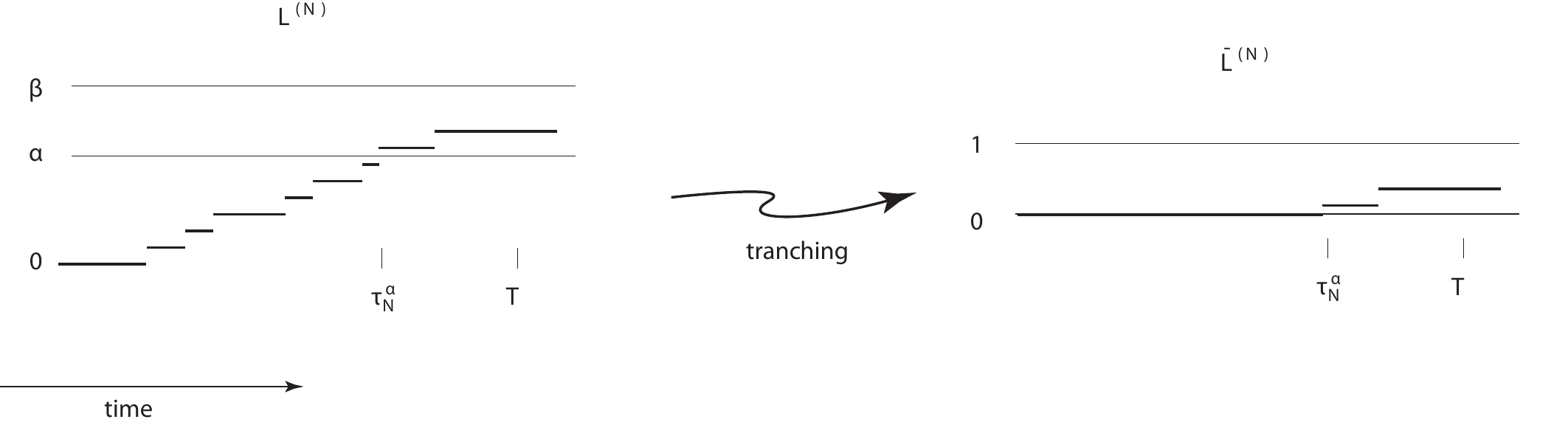}
\caption{Loss processes $L^{(N)}$ and $\tL^{(N)}$}
\label{fig:typ}
\end{figure}

\section{The Model}\label{S:SM}

Let's now think about the sources of randomness in the names.  Each name is affected by its own \emph{idiosyncratic} randomness and by \emph{systemic} randomness (which affects all of the names).  Assumedly, the systemic randomness,
which corresponds to macroeconomic factors, is \emph{low-dimensional}
compared to the number of names.  For example, there may be only a handful of macroeconomic factors which affect a pool of many thousands of names.  We can capture this functionality as
\begin{equation}\label{E:structural} \chi_{\{\tau_n<T\}} = \chi_A(\xi^\idio_n,\xi^\system) \end{equation}
where the $\{\xi^\idio_n\}_{n\in \N}$ and $\xi^\system$ are
all independent random variables, and $A$ is some appropriate set in the product space of the sets where the $\xi^\idio_n$'s and $\xi^\system$ take values.
Since we want the defaults to be identically distributed, we may furthermore assume that the $\xi^\idio_n$'s are identically distributed.

Our interest is to understand the implications of the structural model
\eqref{E:structural}.  We are not so much concerned with specific models for the $\xi^\idio_n$'s, the $\xi^\system$, or the set $A$ but rather the structure of the rare losses in the investment-grade tranches.
We would also like to avoid, as much as possible, a detailed analysis
of the parts of \eqref{E:structural} since in practice what we have available
to carry out pricing calculations is the price of credit default swaps for
the individual names; i.e. (after a transformation), $\BP\{\tau_N<T\}$.
Thus we can't with certainty get our hands on the details of \eqref{E:structural}.  There may in fact be several models of the type \eqref{E:structural}
which lead to the same ``price'' for the rare events involved in an
investment-grade tranche.  If we can understand more about the structure
of rare events in these tranches, we can understand which aspects of
\eqref{E:structural} are important (and then try to calibrate specific models
using that insight).

Regardless of the details of \eqref{E:structural}, we can make some headway.
The notional loss at time $T-$ will be given by
\begin{equation*} L^{(N)}_{T-} = \frac{1}{N}\sum_{n=1}^N \chi_A(\xi^\idio_n,\xi^\system). \end{equation*}
The definition of an investment-grade tranche is that $\BP\lb L^{(N)}_{T-}>\alpha\rb$ is small.  Guided by Chebychev's inequality, lets' define
\begin{equation*} \mu^{(N)} \Def \frac{1}{N}\sum_{n=1}^N \BE\left[\chi_A(\xi^\idio_n,\xi^\system)\right] \qquad \text{and}\qquad \sigma^{(N)} \Def \sqrt{\BE\left[\left(L^{(N)}_{T-}-\mu^{(N)}\right)^2\right]}. \end{equation*}
If $\alpha>\mu^{(N)}$, Chebychev's inequality gives us that
\begin{equation*} \BP\lb L^{(N)}_{T-}>\alpha\rb \le \frac{\left(\sigma^{(N)}\right)^2}{\left(\alpha-\mu^{(N)}\right)^2}. \end{equation*}
In order for this to be small, we would like that $\sigma^{(N)}$ be small;
this is the point of pooling.  For any fixed value of $x$, the conditional
law of $L^{(N)}_{T-}$ given that $\xi^\system=x$ is the variance of 
$\tfrac{1}{N}\sum_{n=1}^N\chi_A(\xi^\idio_n,x)$;
thus the conditional variance of $L^{(N)}_{T-}$ given that $\xi^\system=x$
is at most of order $\tfrac{1}{4N}$.  Hopefully, when we reinsert
the systemic randomness, the variance of $L^{(N)}$ will still be small,
and we will indeed have an investment-grade tranche.

In fact, we can do better than Chebychev's inequality.
By again conditioning on $\xi^\system$, we can write that
\begin{equation*} \BP\lb L^{(N)}_{T-}>\alpha\rb = \BE\left[\BP\lb L^{(N) }_{T-}>\alpha\big|\xi^\system\rb \right] \end{equation*}
Thus the tranche will be investment-grade if $\BP\lb L^{(N) }_{T-}>\alpha\big|\xi^\system=x\rb$ is small for ``most'' values of $x$ (see Remark \ref{R:systemic}).  As mentioned above,
however, we know the law of $L^{(N)}_{T-}$ conditioned on $\xi^\system$.
Namely, 
\begin{equation*} \BP\lb L^{(N) }_{T-}>\alpha\big|\xi^\system=x\rb = 
\BP\lb \frac{1}{N}\sum_{n=1}^N\chi_A(\xi^\idio_n,x)>\alpha\rb. \end{equation*}
This then clearly motivates a natural two-step approach.
Our first step is to condition on the value of the systemic randomness 
(which we may think of as fixing a ``state of the world'' or a ``regime'') and concentrate
on how rare events occur due to idiosyncratic randomness (i.e., to effectively \emph{suppress} the systemic randomness).  It will turn out that
this is in itself a fairly involved calculation.  Nevertheless, it is
connected with a classic problem in large deviations theory---\emph{Sanov's theorem}.  With this in hand, we should
then be able to return to the original problem and average over the systemic randomness (in Subsection \ref{S:Correlated}).  Some of the
finer details of these effects of correlation will appear in sequels
to this paper.  Here we will restrict our interest in the effects of correlation
to a very simple model (which is hopefully nevertheless illustrative).  

Define $I\Def [0,\infty]$ and endow $I$ with its usual topology under which it is Polish and its usual ordering
\footnote{We endow $I$ with the usual topology and ordering.
$I$ is the collection of nonnegative real numbers and a non-real ``point'',
which we label as $\infty$.  Define $\wp:[0,\pi/2]\to I$ as $\wp(t) \Def \tan(t)$ for $t\in [0,\pi/2)$, and define $\wp(\pi/2)\Def \infty$.  Then $\wp$
is a bijection.  The topology
and ordering of $I$ is that given by pushing the topology and ordering
of $[0,\pi/2]$ forward through $\wp$.  Thus $I$
is Polish and in fact compact.}; each of the default times is an $I$-valued random variable.  Since we want to consider a countable
collection of default times, we will take
our event space to be $\Omega\Def I^\N$ and\footnote{As usual, for any topological space $\Xsp$, $\Borel(\Xsp)$ is the Borel sigma-algebra of subsets of $\Xsp$, and $\Pspace(\Xsp)$ is the
collection of probability measures on $(\Xsp,\Borel(\Xsp))$.} we will take $\filt\Def \Borel(I^\N)$.
Fix next $\mu\in \Pspace(I)$; we will want all of the names to be identically distributed with common law $\mu$.
To reflect our initial working assumption
that the names are independent, we now let the risk neutral probability $\BP\in \Pspace(I^\N)$ be defined
by requiring that
\begin{equation*} \BP\left(\bigcap_{n=1}^N\{\tau_n\in A_n\}\right) = \prod_{n=1}^N \mu(A_n). \end{equation*}
for all $N\in \N$ and all $\{A_n\}_{n\in \N}\subset \Borel(I)$.
We also define, in the usual way,
\begin{equation*} F(t) \Def \mu[0,t]. \qquad t\in I \end{equation*}
In principle, one can recover $F$ from prices of credit default swaps.

\begin{example} Our setup includes both the Merton model and the reduced form model.  For the reduced form model, let $\lambda:(0,\infty)$ be the hazard rate and set
\begin{equation*} f(t) = \lambda(t)\exp\left[-\int_{s=0}^t \lambda(s)ds\right] \qquad t\in (0,\infty)\end{equation*}
and let $F$ have density $f$.
On the other hand, for the Merton model with stock volatility $\sigma$, risk-neutral drift $\theta$, initial valuation $1$, and bankruptcy barrier $K\in (0,1)$, we would have
\begin{equation*} f(t) = \frac{\ln (1/K)}{\sqrt{2\pi \sigma^2 t^3}}\exp\left[-\frac{1}{2\sigma^2 t}\left(\left(\theta-\frac{\sigma^2}{2}\right)t+\ln \frac{1}{K}\right)^2\right]. \qquad t\in (0,\infty)\end{equation*}
Again define $F$ by integrating $f$.
\end{example}

We can then rewrite the notional loss process as 
\begin{equation*} L^{(N)}_t = \frac{1}{N}\sum_{n=1}^N\chi_{[0,t]}(\tau_n) = \nu^{(N)}[0,t] \end{equation*}
where $\nu^{(N)}$ is empirical distribution of the $\tau_n$'s; i.e.,
\begin{equation}\label{E:nudef} \nu^{(N)} = \frac{1}{N}\sum_{n=1}^N \delta_{\tau_n}. \end{equation}
We point out that $\nu^{(N)}$ is a random element of $\Pspace(I)$ (i.e., a random measure\footnote{Since the map $x\mapsto \delta_x$ is a measurable map from $I$ to $\Pspace(I)$, each map $\omega\mapsto \delta_{\tau_n(\omega)}$ is a measurable map from $\Omega$ to $\Pspace(I)$.  Thus for each $N$, the map $\omega\mapsto (\delta_{\tau_1(\omega)},\delta_{\tau_2(\omega)}\dots \delta_{\tau_N(\omega)})$ is a measurable map from $\Omega$ to $(\Pspace(I))^N$.  Recalling the definition of the weak
topology as integration against continuous bounded functions, we then see that
the map $(\mu_1,\mu_2\dots \mu_N)\mapsto \frac{1}{N}\sum_{n=1}^N \mu_n$ is continuous and thus measurable as a map from $(\Pspace(I))^N$ to $\Pspace(I)$.  Hence $\nu^{(N)}$ is indeed a $\Pspace(I)$-valued random variable.}).  This formulation
is the starting point for our analysis and will lead to several insights.  In particular, the (weak) law of large numbers implies that for each $t>0$, 
\begin{equation} \label{E:convprob}\lim_{N\to \infty}L^{(N)}_t=F(t). \qquad \text{(in probability)}\end{equation}
More generally, $\nu^{(N)}$ tends to $\mu$ (in the Prohorov topology on $\Pspace(I)$); for every $\eps>0$,
\begin{equation*} \lim_{N\nearrow \infty}\BP\left\{d_{\Pspace(I)}(\nu^{(N)},\mu)\ge \eps\right\}=0. \end{equation*}
where $d_{\Pspace(\R_+)}$ is the Prohorov metric \cite{MR88a:60130}.

Consider now an investment-grade tranche; i.e., a senior or super-senior tranche.  The attachment point for such a tranche should be set so that it is unlikely to suffer any defaults; i.e., it is unlikely that 
$\Prot_N$ is nonzero.   Clearly 
\begin{equation}\label{E:support} \left\{\Prot_N\not = 0\right\} = \{L^{(N)}_{T-}>\alpha\} = \{\nu^{(N)}[0,T)>\alpha\}, \end{equation}
and comparing this with \eqref{E:convprob}, we see that a tranche will be investment-grade if and only an obvious requirement holds:
\begin{assumption}[Investment-grade]\label{A:IG} We assume that
\begin{equation*} \alpha>F(T-).\end{equation*}
\end{assumption}
\noindent In this case, the valuation of such a tranche should depend in large part on how ``rare'' it is that 
$L^{(N)}_{T-}>\alpha$.  As $N$ becomes large, \eqref{E:convprob} means that in fact it becomes less and less likely that $L^{(N)}_{T-}>\alpha$.  Note also that since $\alpha<1$, this assumption implies that $F(T-)<1$.  This is natural; if $F(T-)=1$,
then all defaults must have occurred before $T$, essentially precluding the possibility of constructing an investment-grade tranche.

Combining our comments after \eqref{E:SN} about the structure of $\Prot_N$ and
\eqref{E:support}, we have that
\begin{equation}\label{E:bounded} 0\le \Prot_N\le \chi_{\{\nu^{(N)}[0,T)>\alpha\}}. \end{equation}
Hence for an investment-grade tranche, $\BE[\Prot_N]$ is small if it is unlikely that $\nu^{(N)}[0,T)>\alpha$ (in other words, we don't have any competition between ``big'' values of $\Prot_N$ and ``small'' sets).
Note also that \eqref{E:convprob} implies that $\lim_{N\to \infty}\tL^{(N)}_{T-}=0$ (in probability) so that in fact
\begin{equation}\label{E:Premas}\lim_{N\to \infty} \BE[\Prem_N] = \sum_{t\in \PTimes} e^{-\rate t}. \end{equation}
In other words, if losses are unlikely, all of the premiums will most likely be paid.
Thus the nontrivial part of $S_N$ comes from the protection leg, whose value is small.

Let's now step into the world of large deviations, which tells us how to 
study rare events.  The asymptotics of $\nu^{(N)}$ is exactly the subject of 
\emph{Sanov's} theorem \cite{MR1619036}, which states that $\nu^{(N)}$ has
a large deviations principle with rate function given by relative entropy with respect to $\mu$; i.e., with rate function
\begin{equation*} H(\mu'|\mu) = \begin{cases} \int_{t\in I}\ln \frac{d\mu'}{d\mu}(t)\mu'(dt) &\text{if $\mu'\ll \mu$} \\
\infty &\text{else.}\end{cases}\end{equation*}
Informally, for any $A\in \Borel(\Pspace(I))$,
\begin{equation}\label{E:san} \BP\left\{\nu^{(N)}\in A\right\} \overset{N\nearrow \infty}{\asymp} \exp\left[-N\inf_{\mu'\in A}H(\mu'|\mu)\right]. \end{equation}
Since large deviations is not in the mainstream of financial mathematics (see, however, \cite{MR1644496})
we have summarized some of its foundations in Subsection \ref{S:LD}.  Combining \eqref{E:bounded} with Sanov's theorem, we conjecture that for large $N$
\begin{equation*} \BE\left[\Prot_N\right] \le \BP\left\{\nu^{(N)}[0,T)>\alpha\right\} \overset{N\nearrow \infty}{\asymp} \exp\left[-N \fI(\alpha)\right] \end{equation*}
where
\begin{equation*} \fI(\alpha) \Def \inf\left\{H(\mu'|\mu): \mu'[0,T)\ge \alpha\right\}. \end{equation*}
Although this looks intimidating (it is an infinite-dimensional minimization problem), in fact it has an easy solution and an explicit minimizer.
For $\alpha_1$ and $\alpha_1$ in $[0,1]$, define
\begin{equation*} \hbar(\alpha_1,\alpha_2) \Def \begin{cases} \alpha_1\ln \frac{\alpha_1}{\alpha_2} + (1-\alpha_1)\ln \frac{1-\alpha_1}{1-\alpha_2} &\text{for $\alpha_1$ and $\alpha_2$ in $(0,1)$} \\
\ln \frac{1}{\alpha_2} &\text{for $\alpha_1=1$, $\alpha_2\in (0,1)$} \\
\ln \frac{1}{1-\alpha_2} &\text{for $\alpha_1=0$, $\alpha_2\in [0,1)$} \\
\infty &\text{else.}\end{cases}\end{equation*}
\begin{proposition}\label{P:explicit} We have that
\begin{equation*} \fI(\alpha) = \hbar(\alpha,F(T-)) = H(\tilde \mu^*_\alpha|\mu), \end{equation*}
where
\begin{equation}\label{E:dacc} \tilde \mu^*_\alpha(A) = \mu(A\cap [0,T))\frac{\alpha}{F(T-)} + \mu(A\cap [T,\infty])\frac{1-\alpha}{1-F(T-)} \end{equation}
for all $A\in \Borel(I)$.\end{proposition}
\noindent The proof of this is given Section \ref{S:Proofs}.
In fact, the formula for $\fI$ is what we would expect from considering only $L^{(N)}_{T-}$.  We can think of $L^{(N)}_{T-}$ as counting the normalized number of
heads in a collection of i.i.d. coin flips, where the probability of heads
(i.e., defaults before time $T$)
for each coin is $F(T-)$.  The likelihood that the normalized number of heads
is approximately $\alpha$ is given, via Sanov's theorem, by relative entropy
of a coin flip with bias $\alpha$ with respect to a coin with bias $F(T-)$
(see the comments after Theorem \ref{T:measurechange}).

We are almost ready to state our main theorem.
We need one last assumption.
\begin{assumption}\label{A:density} We assume that $F(T')<F(T)$ for all $T'\in [0,T)$.
\end{assumption}
\noindent In other words, $F$ cannot be flat to the left of $T$.  Thus $F(T)$ is positive (viz., for $T'\in [0,T)$, $F(T)>F(T')\ge 0$); this is natural, since if $F(T)=0$, then there is no possibility of any defaults by time $T$.
Secondly, if $F$ is flat right before $T$, then any defaults by time $T$ must in fact have occurred earlier, so we can effectively reduce the time interval of interest to a smaller one.
By disallowing such a flat, we ensure that there is some likelihood of
defaults right before $T$, allowing us to carry out a quantitative analysis
of $L^{(N)}$ right before time $T$ (see the proof of Lemma \ref{L:TTimes}).

The goal of this paper is to formalize the asymptotics conjectured above.
Set
\begin{equation}\label{E:vkapdef} \vkap\Def \ln \left(\frac{\alpha}{1-\alpha}\frac{1-F(T-)}{F(T-)}\right) = \ln \left(\frac{\frac{1}{F(T-)}-1}{\frac{1}{\alpha}-1}\right). \end{equation}
In light of Assumption \ref{A:IG}, the second formula ensures that $\vkap>0$.
\begin{theorem}[Main]\label{T:Main} We have that
\begin{multline*} \BE\left[\Prot_N\right] = \frac{e^{-\rate T}\exp\left[-\vkap\left(\granup -N\alpha\right)\right]}{N^{3/2}(\beta-\alpha)\sqrt{2\pi\alpha(1-\alpha)}}\lb \frac{\alpha(1-\alpha)F(T-)(1-F(T-))}{(\alpha-F(T-))^2} \right.\\
\left. +\left(\granup-N\alpha\right)\frac{\alpha(1-F(T-))}{\alpha-F(T-)} + \Err(N)\rb \exp\left[-N \fI(\alpha)\right] \end{multline*}
where $\lim_{N\to \infty}\Err(N)=0$.\end{theorem}
\noindent We can recognize a number of effects here.  Firstly, the $e^{-\rate T}$
term reflects the fact that while by assumption losses in the CDO are unlikely, the least unlikely way for them to occur is right before expiry.  The term $\beta-\alpha$ in the denominator reflects the tranche width; note that we are looking at large $N$-approximations here; if we were to first take asymptotics as the tranche
width tends to zero, we would probably capture some different effects
(but we expect that the exponentially small entropy term would still appear).
The $\sqrt{2\pi\alpha(1-\alpha)}$ reflects something like a Gaussian correction
term (it directly comes from the calculations of Section \ref{S:Proofs}).  The $N^{3/2}$
is a combination of two things.  Part of it ($N^{1/2}$) also comes from the Gaussian correction.  The rest ($N$) comes from the actual size of the protection
leg payments $\Prot_N$ once the attachment point has been reached.
The unsightly term $\granup - N\alpha$ comes from an unavoidable granularity in our problem; the loss process can only take on values in $\Z/N$.  We expect this granularity to disappear if the notional loss takes on a continuum of values.  This would be the case, for example, with random recoveries (cf. \cite{AndersonSidenius}).  Of course, by taking $\alpha$ to be a multiple of $1/N$, we can make this granularity disappear---at the cost of making our calculations
look more restrictive than they actually are.

Finally, we explicitly point out that our analysis is \emph{asymptotic} as
the number $N$ of names becomes large.  We cannot say anything specific about
any finite $N$.  This is analogous to the law of large numbers; the law of
large numbers cannot, for example, give information about any finite number
of coin flips, but rather is useful in framing one's thoughts when one
has ``many'' coin flips.

Combining \eqref{E:Premas} and Theorem \ref{T:Main}, we see that the asymptotic behavior of the premium $S_N$ is given by
\begin{equation}\label{E:Sas} \begin{aligned} S_N &= \frac{1}{N^{3/2}}\frac{e^{-\rate T}\exp\left[-\vkap\left(\granup -N\alpha\right)\right]}{\lb \sum_{t\in \PTimes} e^{-\rate t}\rb (\beta-\alpha)\sqrt{2\pi\alpha(1-\alpha)}}\lb \frac{\alpha(1-\alpha)F(T-)(1-F(T-))}{(\alpha-F(T-))^2}\right.\\
&\qquad \left. +\left(\granup-N\alpha\right)\frac{\alpha(1-F(T-))}{\alpha-F(T-)}  + \Err'(N)\rb \exp\left[-N \fI(\alpha)\right] \end{aligned}\end{equation}
where $\lim_{N\to \infty}\Err'(N)=0$.  

To close this section, we plot some ``theoretical'' prices as a function of the number $N$. By ``theoretical'', we mean the
quantity 
\begin{multline*} S^*_N \Def \frac{\exp\left[-\vkap\left(\granup -N\alpha\right)\right]}{N^{3/2}\sqrt{\alpha(1-\alpha)}}\lb \frac{\alpha(1-\alpha)F(T-)(1-F(T-))}{(\alpha-F(T-))^2} \right.\\
\left.+\left(\granup-N\alpha\right)\frac{\alpha(1-F(T-))}{\alpha-F(T-)}\rb \exp\left[-N \fI(\alpha)\right] \end{multline*}
We have here set $\Err'\equiv 0$ in \eqref{E:Sas} and have removed the prefactor
\begin{equation*} \frac{e^{-\rate T}}{\lb \sum_{t\in \PTimes} e^{-\rate t}\rb (\beta-\alpha)\sqrt{2\pi}}. \end{equation*}
\begin{figure}[b]
\includegraphics[width=5in, height=3in]{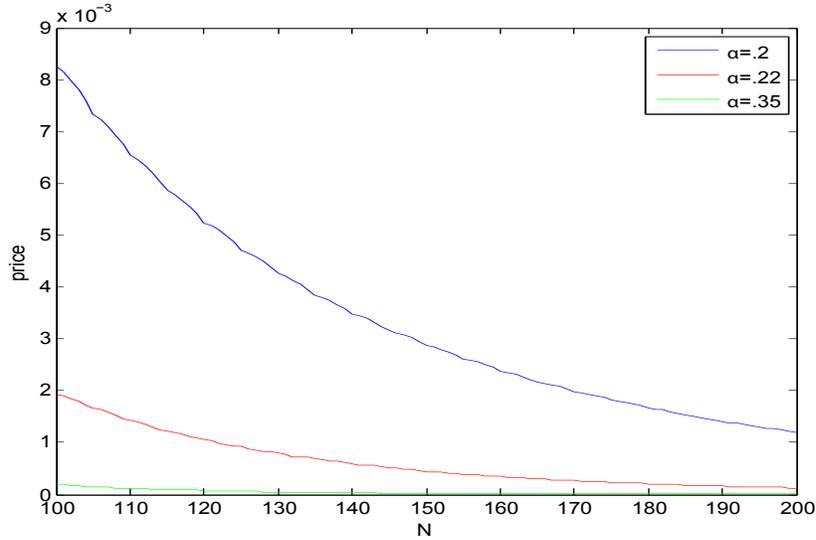}
\caption{$S^*_N$ for several values of $\alpha$}
\label{fig:prices}
\end{figure}

\subsection{Correlation}\label{S:Correlated}
We can now introduce a simple model of correlation without too much trouble.  Assume that $\xi^\system$ takes values in a finite set $\Xsp$.  Fix $\{p(x);\, x\in \Xsp\}$ such that $\sum_{x\in \Xsp}p(x)=1$ and $p(x)>0$ for all $x\in \Xsp$;
we will assume that $\xi^\system$ takes on the value $x$ with probability $p(x)$.  We can think of the set $\Xsp$ as the collection of possible states of the world.  If we believe in \eqref{E:structural}, we should then be in the previous
case if we condition on the various values of $\xi^\system$.  To formalize this,
fix a $\{\mu(\cdot,x)\}_{x\in \Xsp}\subset \Pspace(I)$.  Fix a probability measure $\BP$ such that
\begin{equation}\label{E:corrprob} \BP\left(\bigcap_{n=1}^N\{\tau_n\in A_n\}\right) = \sum_{x\in\Xsp}\lb \prod_{n=1}^N \mu(A_n,x)\rb p(x). \end{equation}
for all $\{A_n\}_{n\in \N}\subset \Borel(I)$.

To adapt the previous calculations to this case, we need the analogue of Assumptions \ref{A:IG} and \ref{A:density}.
Namely, we need that $\max_{x\in \Xsp}\mu([0,T),x)<\alpha$ and also that $\mu([0,T'],x)<\mu([0,T],x)$ for all $T'\in [0,T)$ and all $x\in \Xsp$.

\begin{remark}\label{R:systemic} The requirement that $\max_{x\in \Xsp}\mu([0,T),x)<\alpha$ is a particularly unrealistic one.  It means that the tranche losses will be rare
for \emph{all} values of the systemic parameter.  In any truly applicable
model, the losses will come from a combination of bad values of the
systemic parameter and from tail events in the pool of idiosyncratic randomness
(i.e., we need to balance the size of $\BP\lb L^{(N) }_{T-}>\alpha\big|\xi^\system=x\rb$ against the distribution of $\xi^\system$).
One can view our effort here as study which focusses primarily on tail events in the pool of idiosyncratic randomness.
Any structural model which attempts to study losses due to both idiosyncratic and systemic randomness will most likely involve calculations which are similar
in a number of ways to ours here.  We will explore this issue elsewhere.
\end{remark}

For each $x\in \Xsp$, define
\begin{equation*} \vkap_x\Def \ln \left(\frac{\alpha}{1-\alpha}\frac{1-\mu([0,T),x)}{\mu([0,T),x)}\right) = \ln \left(\frac{\frac{1}{\mu([0,T),x)}-1}{\frac{1}{\alpha}-1}\right). \end{equation*}
Then
\begin{multline*} \BE\left[\Prot_N\right] = \frac{e^{-\rate T}}{N^{3/2}(\beta-\alpha)\sqrt{2\pi\alpha(1-\alpha)}}\sum_{x\in \Xsp}\left(\exp\left[-\vkap_x\left(\granup -N\alpha\right)\right]\right.\\
\times \left.\lb \frac{\alpha(1-\alpha)\mu([0,T),x)\left(1-\mu([0,T),x)\right)}{\left(\alpha-\mu([0,T),x)\right)^2} +\left(\granup-N\alpha\right)\frac{\alpha\left(1-\mu([0,T),x)\right)}{\alpha-\mu([0,T),x)} + \Err(N)\rb \right.\\
\left.\times \exp\left[-N \hbar(\alpha,\mu([0,T),x))\right]p(x)\right) \end{multline*}
where $\lim_{N\to \infty}\Err_x(N)=0$ for all $x\in \Xsp$.  Similarly we have that
\begin{multline*} S_N = \frac{1}{N^{3/2}}\frac{e^{-\rate T}}{\lb \sum_{t\in \PTimes} e^{-\rate t}\rb (\beta-\alpha)\sqrt{2\pi\alpha(1-\alpha)}}\sum_{x\in \Xsp}\left(\exp\left[-\vkap_x\left(\granup -N\alpha\right)\right]\right.\\
\left. \times \lb \frac{\alpha(1-\alpha)\mu([0,T),x)\left(1-\mu([0,T),x)\right)}{\left(\alpha-\mu([0,T),x)\right)^2}+\left(\granup-N\alpha\right)\frac{\alpha\left(1-\mu([0,T),x)\right)}{\alpha-\mu([0,T),x)}  + \Err'(N)\rb\right.\\
\left.\times \exp\left[-N \hbar(\alpha,\mu([0,T),x))\right]p(x)\right) \end{multline*}
where $\lim_{N\to \infty}\Err'_x(N)=0$ for all $x\in \Xsp$.
If we further assume that there is a unique $x^*\in \Xsp$ such that $\min_{x\in \Xsp}  \fI(\alpha,\mu([0,T),x))=\fI(\alpha,\mu([0,T),x^*))$, we furthermore
have that
\begin{align*} \BE\left[\Prot_N\right] &= \frac{e^{-\rate T}}{N^{3/2}(\beta-\alpha)\sqrt{2\pi\alpha(1-\alpha)}}\exp\left[-\vkap_{x^*}\left(\granup -N\alpha\right)\right]\\
&\qquad \times \lb \frac{\alpha(1-\alpha)\mu([0,T),x^*)\left(1-\mu([0,T),x^*)\right)}{\left(\alpha-\mu([0,T),x^*)\right)^2} 
+\left(\granup-N\alpha\right)\frac{\alpha\left(1-\mu([0,T),x^*)\right)}{\alpha-\mu([0,T),x^*)} + \Err(N)\rb \\
&\qquad \times \exp\left[-N \hbar(\alpha,\mu([0,T),x^*))\right]p(x^*)\\
S_N &= \frac{1}{N^{3/2}}\frac{e^{-\rate T}}{\lb \sum_{t\in \PTimes} e^{-\rate t}\rb (\beta-\alpha)\sqrt{2\pi\alpha(1-\alpha)}}\exp\left[-\vkap_{x^*}\left(\granup -N\alpha\right)\right]\\
&\qquad \times \lb \frac{\alpha(1-\alpha)\mu([0,T),x^*)\left(1-\mu([0,T),x^*)\right)}{\left(\alpha-\mu([0,T),x^*)\right)^2}+\left(\granup-N\alpha\right)\frac{\alpha\left(1-\mu([0,T),x^*)\right)}{\alpha-\mu([0,T),x^*)}  + \Err'(N)\rb\\
&\qquad \times \exp\left[-N \hbar(\alpha,\mu([0,T),x^*))\right]p(x^*) \end{align*}
where $\lim_{N\to \infty}\Err(N)=0$ and $\lim_{N\to \infty}\Err'(N)=0$.

Note that we can use this methodology to approximately study Gaussian
correlations.  Fix a positive $M\in \N$ and define $x_i\Def \tfrac{i}{M}$
for $i\in \{-M^2,-M^2+1\dots M^2\}$; set $\Xsp \Def \{x_i\}_{i=-M^2}^{M^2}$.
Define
\begin{equation*} \Phi(x) \Def \int_{t=-\infty}^x\frac{1}{\sqrt{2\pi}}\exp\left[-\frac{t^2}{2}\right]dt \qquad x\in \R \end{equation*}
as the standard Gaussian cumulative distribution function.
Define
\begin{equation*} p(x_i) \Def \begin{cases} \Phi\left(x_i+\frac{1}{2M}\right)-\Phi\left(x_i-\frac{1}{2M}\right) &\qquad \text{if $i\in \{-M^2+1,\dots M^2-1\}$} \\
\Phi\left(x_{-M^2}+\frac{1}{2M}\right) &\qquad \text{if $i=-M^2$}\\
1-\Phi\left(x_{M^2}-\frac{1}{2M}\right) &\qquad \text{if $i=M^2$}\end{cases}\end{equation*}
If we have a pool of $N$ names with common probability of default $p$ by time $T$ and we want to consider a Gaussian copula with
correlation $\rho>0$ (the case $\rho<0$ can be dealt with similarly),
we would take the $\mu(\cdot,x_i)$'s such that
\begin{equation*} \mu([0,T),x_i) \Def \Phi\left(\frac{\Phi^{-1}(p)-\rho x_i}{\sqrt{1-\rho^2}}\right). \end{equation*}
This is related to the calculations of \cite{Glasserman} and \cite{MR2384674};
those calculations are asymptotically related to our calculations.  We shall explore the connection with these two papers elsewhere.  We note, by way of contrast with \cite{Glasserman} and \cite{MR2384674}, that our efforts give a good picture of the \emph{dynamics} of the loss process prior to expiry.
We also note that our model of \eqref{E:corrprob} is entirely comfortable with non-Gaussian correlation.  Note also that one could also (by discretization) allow the systemic parameter $\xi^\system$ to be path-valued.

\subsection{Large Deviations}\label{S:LD}
We shall here give a very short summary of the main ideas of large deviations; see \cite{MR1619036} for a comprehensive treatment.  The basic observation behind
the theory is that a sum of exponentials behaves like largest-growing
exponential.  For example,
\begin{equation*} e^{-3N} + e^{7N} + e^{4N} = e^{7N}\lb 1 + e^{-10N} + e^{-3N}\rb \asymp e^{-7N}. \end{equation*}
Here ``$\asymp$'' means ``having the same exponential growth''; in other words,
$A_N\asymp B_N$ if $\lim_{N\to \infty}\tfrac1N\ln A_N=\lim_{N\to \infty}\tfrac1N\ln B_N$.
\emph{Laplace asymptotics} extends this to integrals.  This is a relevant place
to start the study of rare events if we consider a collection $\{X_n\}_{n\in \N}$
of random variables whose laws are of the form
\begin{equation}\label{E:laplace} \BP\{X_N\in A\} \Def \int_{x\in A}c_N\exp\left[-N \phi(x)\right]dx \qquad A\in \Borel(\R) \end{equation}
for some $\phi\in C(\R)$ and some normalization constant $c_N$ (e.g., if we take $\phi(x) = \tfrac12(x-1)^2$ and $c_N = 1/\sqrt{2\pi/N}$, then $X_N$ will be a
normal random variable with mean $1$ and variance $\tfrac{1}{N}$).  If we
assume that $\phi$ has nice enough growth properties (so that the integrals
in \eqref{E:laplace} are well-defined and $c_N$ has subexpontial growth) ,
then Laplace asymptotics states that
\begin{equation}\label{E:expasymp} \BP\{X_N\in A\} \asymp \exp\left[-N \inf_{x\in A}\phi(x)\right] \end{equation}
for ``nice'' enough sets $A$.  By taking $A=\R$, we see that
we must have that $\inf_{x\in \R}\phi(x)=0$.  If this minimum is achieved at a single point $x^*$, then by taking $A$ as the complement of a neighborhood of $x^*$
we have that $X_N\to x^*$ in probability, so $\{X_N\in A\}$ is a rare
event for any nice enough set $A$ not containing $x^*$.

One of the main aspects of large deviations theory is something of an inverse problem.  Can we have \eqref{E:expasymp} even without \eqref{E:laplace}?
In some cases, yes.  Fix $\theta\in \R$
and consider the limiting rate of growth of the logarithmic moment
generating function; we have that
\begin{equation*} \lim_{N\to \infty}\frac1N\ln \BE\left[\exp\left[\theta N X_N\right]\right] = \lim_{N\to \infty}\frac1N\ln \int_{x\in \R}c_N\exp\left[N\lb \theta x-\phi(x)\rb\right]dx = \sup_{x\in \R}\lb \theta x-\phi(x)\rb. \end{equation*}
The key realization is that the right-hand side is the \emph{Legendre-Fenchel transform} of $\phi$, and that if $\phi$ has nice convexity properties,
we can recover $\phi$ from $M$ by taking the Legendre-Fenchel transform
again; i.e., 
\begin{equation*} \phi(x) = \sup_{\theta\in \R}\lb \theta x-M(\theta)\rb. \end{equation*}
The strength of this chain of arguments is that the moment generating function is
well-defined (but of course possibly infinite) regardless of whether
$X_N$ is discrete or continuous.  It even makes sense when $X_N$ 
takes values in an infinite-dimensional topological linear space $\Xsp$
if we replace multiplication by $\theta$ with the action of a linear functional
on $\Xsp$.  The rigorous definition of a large deviations principle is as follows \cite{MR758258}.  We say that $\{X_n;\, n\in \N\}$ (which we now assume
to take values in a topological space $\Xsp$) has a large deviations principle with rate function $\calI:\Xsp\to [0,\infty]$ if the following three requirements hold:
\begin{itemize}
\item For every $s\ge 0$, $\{x\in \Xsp: \calI(x)\le s\}$ is a compact subset of $\Xsp$.
\item For every open subset $G$ of $\Xsp$,
\begin{equation*} \varliminf_{N\to \infty}\frac{1}{N}\ln \BP\{X_n\in G\} \ge - \inf_{x\in G}\calI(x). \end{equation*}
\item For every closed subset $F$ of $\Xsp$,
\begin{equation*} \varlimsup_{N\to \infty}\frac{1}{N}\ln \BP\{X_n\in F\} \le - \inf_{x\in F}\calI(x). \end{equation*}
\end{itemize}

Returning to our focus, which is Sanov's theorem applied to \eqref{E:nudef},
we have that for any $\phi\in C(I)$ (the dual of $\Pspace(I)$),
\begin{equation*} \lim_{N\to \infty}\frac1N\ln \BE\left[\exp\left[N\int_{t\in I}\phi(t)\nu^{(N)}(dt)\right]\right] = \ln \int_{t\in I}e^{\phi(t)}\mu(dt) \end{equation*}
and we can then show that
\begin{equation}\label{E:dual} H(\mu'|\mu) = \sup_{\phi\in C(I)}\lb \int_{t\in I}\phi(t)\mu'(dt) - \ln \int_{t\in I}e^{\phi(t)}\mu(dt)\rb. \qquad \mu'\in \Pspace(I)\end{equation}
This suggests that indeed we should have \eqref{E:san} as interpreted as a large
deviations principle (Sanov's theorem).

\section{A Measure Transformation}\label{S:Tilt}
One of the things which naturally occurs in proofs of large deviations
principles is a \emph{measure change} under which the unlikely event becomes more likely---the cost
of this change of measure is exactly the desired exponential rate of decay (see \cite{MR1619036}).
Let's see what this looks like in our situation (see \cite{MR1619036} for a more
complete motivation of measure changes in large deviations).  Define
\begin{equation*} \phi^*_\alpha(t) = \ln \frac{d\tilde \mu^*_\alpha}{d\mu}(t) = \ln \frac{\alpha}{F(T-)}\chi_{[0,T)}(t) + \ln \frac{1-\alpha}{1-F(T-)}\chi_{[T,\infty]}(t)\qquad t\in I \end{equation*}
(note that since $\fI(\alpha)<\infty$, $H(\tilde \mu^*_\alpha|\mu)<\infty$, so $\tilde \mu^*_\alpha \ll \mu$).  
It is easy to verify that
\begin{equation*} H(\tilde \mu^*_\alpha|\mu) = \int_{t\in I}\phi^*_\alpha(t)\tilde \mu^*_\alpha(dt) - \ln \int_{t\in I}e^{\phi^*_\alpha(t)}\mu(dt); \end{equation*}
thus $\phi^*_\alpha$ is the extremal in the variational representation \eqref{E:dual} for $H(\tilde \mu^*_\alpha|\mu)$ (if we allow ourselves
to extend the supremum over $C(I)$ to the collection of bounded measurable functions; it turns out that this is allowable).
In our analysis of $\nu^{(N)}$ of \eqref{E:nudef}, $\phi^*_\alpha$
will naturally give us an optimal way to ``tilt'' our original
probability measure so that it becomes likely that $\nu^{(N)}[0,T)\approx \alpha$.  The penalty for doing this is exactly $\fI(\alpha)$.
\begin{theorem}\label{T:measurechange}  We have that
\begin{equation*} \BE[\Prot_N] = I_Ne^{-N\fI(\alpha)} \end{equation*}
for all positive integers $N$, where
\begin{equation}\label{E:IDef} I_N\Def \tilde \BE_N\left[\Prot_N\exp\left[-\vkap\gamma_N\right]\chi_{\{\gamma_N>0\}}\right] \end{equation}
where in turn
\begin{equation}\label{E:MIX}\begin{aligned} \tilde \BP_N(A) &\Def \BE\left[\chi_A\prod_{n=1}^N \frac{d\tilde \mu^*_\alpha}{d\mu}(\tau_n)\right] \qquad A\in \filt \\
\gamma_N&= \sum_{n=1}^N \lb \chi_{[0,T)}(\tau_n)-\alpha\rb =N(L_{T-}^{(N)}-\alpha)\end{aligned}\end{equation}
Under $\tilde \BP_N$, $\{\tau_1,\tau_2\dots \tau_N\}$ are independent and identically distributed with common law $\tilde \mu^*_\alpha$.
\end{theorem}
\begin{proof}  Set
\begin{equation*}
\Gamma_N =  N\lb \int_{t\in I}\phi^*_\alpha(t)\nu^{(N)}(dt)-\int_{t\in I}\phi^*_\alpha(t)\tilde \mu^*_\alpha(dt)\rb\end{equation*}
Then
\begin{equation*} \BE[\Prot_N] = \frac{\BE\left[\Prot_N \exp\left[-\Gamma_N\right]\exp\left[\Gamma_N\right]\right]}{\BE\left[\exp\left[\Gamma_N\right]\right]}\BE\left[\exp\left[\Gamma_N\right]\right]. \end{equation*}
Note that
\begin{align*} \int_{t\in I}\phi^*_\alpha(t)\tilde \mu^*_\alpha(dt)&= \fI(\alpha)\\
\exp\left[N \int_{t\in I}\phi^*_\alpha(t)\nu^{(N)}(dt)\right]&= \exp\left[\sum_{n=1}^N\ln \frac{d\tilde \mu^*_\alpha}{d\mu}(\tau_n)\right] = \prod_{n=1}^N \frac{d\tilde \mu^*_\alpha}{d\mu}(\tau_n) \\
\BE\left[\exp\left[\Gamma_N\right]\right]&= e^{-N\fI(\alpha)}\BE\left[\prod_{n=1}^N \frac{d\tilde \mu^*_\alpha}{d\mu}(\tau_n)\right] = e^{-N\fI(\alpha)} \end{align*}
(these equalities in fact reflect some of the basic properties of
large deviations measure transformations and are intimately related with
the fact that $\phi^*_\alpha$ solves the variational problem \eqref{E:dual} associated with $H(\tilde \mu^*_\alpha|\mu)$).
We also clearly have that
\begin{equation*}\frac{\BE\left[\chi_A\exp\left[\Gamma_N\right]\right]}{\BE\left[\exp\left[\Gamma_N\right]\right]}\\
= \frac{\BE\left[\chi_A\exp\left[N\int_{t\in I}\phi^*_\alpha(t)\nu^{(N)}(dt)\right]\right]}{\BE\left[\exp\left[N\int_{t\in I}\phi^*_\alpha(t)\nu^{(N)}(dt)\right]\right]} = \tilde \BP_N(A) \end{equation*}
for all $A\in \filt$.  The properties of $\tilde \BP_N$ are clear from the explicit formula.
We next check that
\begin{multline*}\Gamma_N=N\lb \ln \frac{\alpha}{F(T-)} \nu^{(N)}[0,T) + \ln \frac{1-\alpha}{1-F(T-)}\nu^{(N)}[T,\infty] 
-\frac{\alpha}{F(T-)} \alpha - \ln \frac{1-\alpha}{1-F(T-)}(1-\alpha)\rb\\
=N\ln \frac{\alpha}{F(T-)} \lb \nu^{(N)}[0,T)-\alpha\rb +N \ln \frac{1-\alpha}{1-F(T-)}\lb \nu^{(N)}[T,\infty]-(1-\alpha)\rb = \vkap\gamma_N. \end{multline*}
Finally, we see that $\Prot_N$ is nonzero only if $\gamma_N>0$; we have explicitly included this
in the expression for $I_N$.
\end{proof}
\noindent We note here that
\begin{equation*} \tilde \BE_N\left[L^{(N)}_{T-}\right] = \tilde \mu^*_\alpha[0,T)=\alpha\qquad \text{and}\qquad \tilde \BE_N\left[\left(L^{(N)}_{T-}-\alpha\right)^2\right] = \frac{\alpha(1-\alpha)}{N^2}\le \frac{1}{4N^2}, \end{equation*}
so by Chebychev's inequality, we have that 
\begin{equation*} \lim_{N\to \infty}\tilde \BP_N\lb \left|L^{(N)}_{T-}-\alpha\right|\ge \eps\rb = 0 \end{equation*}
for every $\eps>0$.  In other words, $L^{(N)}_{T-}$ tends to the attachment point $\alpha$ under the sequence $(\tilde \BP_N)_{N\in \N}$ of probability measures
and thus loss is not a rare event under $\tilde \BP_N$ as $N\nearrow \infty$.

We also note that we need to understand the appropriate
change of measure for the empirical measure $\nu^{(N)}$ (as opposed to the
change of measure for the empirical sum $L^{(N)}_{T-}$) since $\Prot_N$ involves the dynamics of the loss process (and not just the probability of loss).

\section{Asymptotic Analysis}
Where do we now stand?  If we can show that $I_N$ has no exponential growth or decay (comparable to $e^{-N\fI(\alpha)}$) then we have successfully identified the asymptotic behavior of $\BE[\Prot_N]$; we will have decomposed it into
an exponentially small part and a prefactor which is of order 1 as $N\nearrow \infty$.  Our goal now is to organize our thoughts about the prefactor,
and in particular to actually extract the asymptotics of Theorem \ref{T:Main};
i.e., to ``do the math''.

Looking at the expression \eqref{E:IDef} for $I_N$, we see that the dominant part of $I_N$ will be where $\gamma_N$ is order\footnote{actually, it will be where $\gamma_N\ll \sqrt{N}$} 1; if $\gamma_N\gg 1$, then $\exp[-\vkap \gamma_N]$ will be very small so the contribution to $I_N$ will be negligible (recall here that $\Prot_N$ is bounded).  This suggests
we organize the formula for $I_N$ based on the values of $\gamma_N$.
Note that the range of $\gamma_N$ when it is positive is
$\CS_N\Def \{n-N\alpha: \text{$n\in \Z$ and $N\alpha\le n\le N$}\}$.
\begin{definition} For each $N$, let $H_N:\CS_N\to [0,1]$ be such that
\begin{equation*} H_N(\gamma_N)= \tilde \BE_N\left[\Prot_N\big|\gamma_N\right] \end{equation*}
on $\{\gamma_N>0\}$.
\end{definition}
\noindent Then we have that
\begin{equation*} I_N = \tilde \BE_N\left[H_N(\gamma_N)\chi_{\{\gamma_N>0\}}\exp\left[-\vkap \gamma_N\right]\right]. \end{equation*}
\noindent It turns out that $H_N$ has very nice asymptotics.
\begin{lemma}\label{L:hasymp} For all $N$, we have that
\begin{equation*} H_N(s) = \frac{e^{-\rate T}s\lb 1 + \Err_1(s,N)\rb}{(\beta-\alpha)N} \end{equation*}
where
\begin{equation*} \varlimsup_{N\nearrow \infty}\sup_{\substack{s\in \CS_N \\ s\le N^{1/4}}}|\Err_1(s,N)|=0. \end{equation*}
\end{lemma}
\noindent We will prove this in Section \ref{S:HAS}.

The next step is to understand the distribution of $\gamma_N$.
\begin{lemma}\label{L:probasymp} We have that
\begin{equation*} \tilde \BP_N\{\gamma_N=s\} = \frac{1+\Err_2(s,N)}{\sqrt{2\pi N\alpha(1-\alpha)}} \end{equation*}
for all $N$ and all $s\in \CS_N$, where
\begin{equation*} \varlimsup_{N\nearrow \infty}\sup_{\substack{s\in \CS_N \\ s\le N^{1/4}}}|\Err_2(s,N)|=0. \end{equation*}
\end{lemma}
\noindent We will prove this in Section \ref{S:Proofs}.  Using this result, we can now start our proof of Theorem \ref{T:Main}.  Set
\begin{align*} \tilde I_{1,N} &\Def \sum_{\substack{s\in \CS_N\\s\le N^{1/4}}}(s-\alpha) e^{-\vkap s}\\
\tilde I_{2,N} &\Def \exp\left[-\vkap\left(\granup -N\alpha\right)\right]\lb \frac{e^{-\vkap}}{(1-e^{-\vkap})^2} +\frac{\granup-N\alpha}{1-e^{-\vkap}}\rb \end{align*}
We thus expect that
\begin{equation*} I_N \approx \frac{e^{-\rate T}\tilde I_{1,N}}{N^{3/2}(\beta-\alpha)\sqrt{2\pi \alpha(1-\alpha)}}. \end{equation*}
We then claim that $\tilde I_{1,N} \approx \tilde I_{2,N}$.  As a preliminary
to showing this, let's recall some calculations about geometric series.
For $\lambda>0$ and each positive integer $n$,
\begin{equation*} \sum_{j=0}^n e^{-\lambda j}= \frac{1}{1-e^{-\lambda}}- \frac{e^{-\lambda(n+1)}}{1-e^{-\lambda}}. \end{equation*}
Differentiating with respect to $\lambda$, we get that
\begin{equation*} \sum_{j=0}^n j e^{-\lambda j}  = \frac{e^{-\lambda}}{(1-e^{-\lambda})^2} -e^{-\lambda(n+1)}\frac{n(1-e^{-\lambda}) + 1}{(1-e^{-\lambda})^2}. \end{equation*}
Let's bound the error terms in these expressions.  Note that $\sup_{x>0}x e^{-x} = e^{-1}$.  For $\lambda>0$ we have that
\begin{align*} \left|e^{-\lambda(n+1)}\frac{n(1-e^{-\lambda}) + 1}{(1-e^{-\lambda})^2}\right|
&\le e^{-\lambda(n+1)}\frac{n+1}{(1-e^{-\lambda})^2}\\
&= 2\lb \frac{\lambda}{2} (n+1)\exp\left[-\frac{\lambda}{2}(n+1)\right]\rb \frac{\exp\left[-\frac{\lambda}{2}(n+1)\right]}{\lambda\left(1-e^{-\lambda}\right)^2} \\
&\le 2e^{-1} \frac{\exp\left[-\frac{\lambda}{2}(n+1)\right]}{\lambda\left(1-e^{-\lambda}\right)^2} \end{align*}
and similarly
\begin{equation*} \left|\frac{e^{-\lambda(n+1)}}{1-e^{-\lambda}}\right|
= 2 e^{-\lambda n/2}\left(1-e^{-\lambda}\right)\lb \frac{\lambda}{2}e^{-\lambda/2}\rb \frac{\exp\left[-\frac{\lambda}{2}(n+1)\right]}{\lambda\left(1-e^{-\lambda}\right)^2}\\
\le 2 e^{-1}\frac{\exp\left[-\frac{\lambda}{2}(n+1)\right]}{\lambda\left(1-e^{-\lambda}\right)^2}. \end{equation*}
Observe now that
\begin{equation*} \lfloor N\alpha + N^{1/4}\rfloor - \granup + 1
\ge N\alpha + N^{1/4} -1-N\alpha -1+1 =N^{1/4}-1  \end{equation*}
for all $N\in \N$.
Combining things and recalling that $\vkap>0$, we see that for all $N\in \N$,
\begin{equation} \label{E:QQ} \begin{aligned} \tilde I_{1,N} &= \sum_{\substack{j\in \Z\\0\le j-N\alpha \le N^{1/4}}}(j-N\alpha)\exp\left[-\vkap (j-N\alpha)\right]\\
&= \sum_{j=\granup}^{\lfloor N\alpha + N^{1/4}\rfloor}(j-N\alpha)\exp\left[-\vkap (j-N\alpha)\right]\\
&= \sum_{j=0}^{\lfloor N\alpha +N^{1/4}\rfloor-\granup}(j+\granup-N\alpha)\exp\left[-\vkap (j+\granup-N\alpha)\right]\\
&= \exp\left[-\vkap (\granup-N\alpha)\right]\lb \sum_{j=0}^{\lfloor N\alpha +N^{1/4}\rfloor-\granup}j e^{-\vkap j} + \left(\granup-N\alpha\right)\sum_{j=0}^{\lfloor N\alpha +N^{1/4}\rfloor-\granup}e^{-\vkap j}\rb \\
&= \exp\left[-\vkap(\granup-N\alpha)\right]\lb \frac{e^{-\vkap}}{(1-e^{-\vkap})^2} + \left(\granup-N\alpha\right)\frac{1}{1-e^{-\vkap}} + \Err_3(N)\rb
\end{aligned}\end{equation}
where
\begin{equation}\label{E:QQQ} |\Err_3(N)| \le 4e^{-1}\frac{\exp\left[-\frac{\vkap}{2}(N^{1/4}-1)\right]}{\vkap \left(1-e^{-\vkap}\right)^2}. \end{equation}
As a consequence, we furthermore have that
\begin{multline*}\left|\tilde I_{1,N}\right|
\le \frac{e^{-\vkap}}{\left(1-e^{-\vkap}\right)^2} + \frac{1}{1-e^{-\vkap}} + 4e^{-1}\frac{\exp\left[-\frac{\vkap}{2}(N^{1/4}-1)\right]}{\vkap \left(1-e^{-\vkap}\right)^2}\\ 
\le \frac{\vkap\left(e^{-\vkap}+1-e^{-\vkap}\right) + 4e^{-1}}{\vkap\left(1-e^{-\vkap}\right)^2}
\le \frac{4e^{-1}(1+\vkap)}{\vkap\left(1-e^{-\vkap}\right)^2}. \end{multline*}
From \eqref{E:vkapdef}, we have that
\begin{equation*} e^{-\vkap} = \frac{1-\alpha}{\alpha}\frac{F(T-)}{1-F(T-)}\qquad \text{and}\qquad 1-e^{-\vkap}= \frac{\alpha-F(T-)}{\alpha(1-F(T-))}. \end{equation*}
so
\begin{equation*} \frac{e^{-\vkap}}{(1-e^{-\vkap})^2} = \frac{1-\alpha}{\alpha}\frac{F(T-)}{1-F(T-)}\frac{\alpha^2(1-F(T-))^2}{(\alpha-F(T-))^2} = \frac{\alpha(1-\alpha)F(T-)(1-F(T-))}{(\alpha-F(T-))^2}. \end{equation*}

We can finally prove our desired result.
\begin{proof}[Proof of Theorem \ref{T:Main}]
We have that
\begin{equation*} I_N = \frac{\tilde I_{2,N}}{N^{3/2}(\beta-\alpha)\sqrt{2\pi \alpha(1-\alpha)}} + \sum_{j=1}^5 \tilde \Err_j(N)\end{equation*}
where
\begin{align*}
\tilde \Err_1(N)&\Def \tilde \BE_N\left[\Prot_N e^{-\vkap \gamma_N}\chi_{\{\gamma_N>N^{1/4}\}}\right] \\
\tilde \Err_2(N)&\Def \sum_{\substack{s\in \CS_N\\s\le N^{1/4}}}\frac{H_N(s)e^{-\vkap s}\Err_2(s,N)}{\sqrt{2\pi N\alpha(1-\alpha)}} \\
\tilde \Err_3(N)&\Def \frac{e^{-\rate T}}{\beta-\alpha}\sum_{\substack{s\in \CS_N\\s\le N^{1/4}}}\frac{se^{-\vkap s}\Err_1(s,N)}{N^{3/2}\sqrt{2\pi \alpha(1-\alpha)}}\\
\tilde \Err_4(N)&\Def \frac{e^{-\rate T}}{\beta-\alpha}\frac{\exp\left[-\vkap\left(\granup-N\alpha\right)\right]\Err_3(N)}{N^{3/2}\sqrt{2\pi \alpha(1-\alpha)}}\end{align*}
Then there is a $\KK_1>0$ such that 
\begin{equation*} |\tilde \Err_1(N)|\le \frac{1}{\KK_1}e^{-\KK_1 N^{1/4}}\qquad \text{and}\qquad |\tilde \Err_4(N)|\le \frac{1}{\KK_1}e^{-\KK_1 N^{1/4}}\end{equation*}
for all $N\in \N$.
Furthermore, we can fairly easily see that there is a $\KK_2$ such that
\begin{equation*}
|\tilde \Err_2(N)|\le \frac{\KK_2\tilde I_{1,N}}{N^{3/2}} \sup_{\substack{s\in \CS_N \\ s\le N^{1/4}}}|\Err_2(s,N)|\qquad \text{and}\qquad |\tilde \Err_3(N)|\le \frac{\KK_2 \tilde I_{1,N}}{N^{3/2}}\sup_{\substack{s\in \CS_N \\ s\le N^{1/4}}}|\Err_1(s,N)| \end{equation*}
for all $N\in \N$ (note from \eqref{E:QQ} and \eqref{E:QQQ} that $\tilde I_{1,N}$ is uniformly bounded in $N$).  Combine things together to get the stated result.\end{proof}

\begin{remark}\label{R:Comments}  Several comments are in order about the analysis of this section.

Firstly, we re-emphasize that we first identified the law of $L^{(N)}_{T-}$ and then studied the law of $L^{(N)}$ right before $T$.
For investment-grade tranches, only this last part of $L^{(N)}$ should be
of interest. For an investment-grade tranche, losses in general should be
rare events; losses significantly before expiry should be \emph{very} rare
events.  This would follow from a detailed analysis of the measure transformation
of Section \ref{S:Tilt}.

Secondly, our analysis here suggests that in more realistic models
(i.e., not i.i.d. names), the first order of business should be a thorough
study of the law of $L^{(N)}_{T-}$.  This is somewhat appealing; by time $T$,
various transients will assumedly have died out, and some sort of macroscopic
analysis may be available.

The third point of interest is the asymptotics of Lemma \ref{L:probasymp}.
This does \emph{not} directly reflect a Poisson distribution for $L^{(N)}_{T-}$.
A number of other studies of CDO's have modelled the loss process as a Poisson
process; an interesting question would thus be to try to find a limiting regime of our calculations which leads to Poisson statistics.
\end{remark}

Finally, it would not be hard to use the measure change of Section \ref{S:Tilt}
and calculations similar to those of this section
to compute the expected loss given default.  We will leave that to the reader.

\section{Proof of Lemma \ref{L:hasymp}}\label{S:HAS}
We here prove Lemma \ref{L:hasymp}.  To do so, we need to develop a clear picture of the dynamics of $L^{(N)}$.  We note that the calculations of this section, though technical, provide a \emph{direct} link to the distribution of the default times.

First of all, we recall that the definition of $\tL^{(N)}$ implies that $\tL^{(N)}$ is nonzero only where $L^{(N)}$ exceeds $\alpha$; since $L^{(N)}$ is nondecreasing, this will in fact be an interval.  Set
\begin{align*} \tau^\alpha_N &\Def \inf\{r>0: \tL^{(N)}_r>0\} = \inf\{r>0: L^{(N)}_r>\alpha\} \\
\tau^\beta_N &\Def \sup\{r>0: \tL^{(N)}_r<\beta-\alpha\} = \sup\{r>0: L^{(N)}_r<\beta\} \end{align*}
A typical graph of $\tL^{(N)}$ is given in Figure \ref{fig:typ}.
Next note that on $\{\gamma_N>0\}$,
\begin{equation}\label{E:AA} \Prot_N = \int_{s\in [\tau^\alpha_N,\tau^\beta_N]\cap[0,T)}e^{-\rate s}d\tL^{(N)}_s. \end{equation}
If $\gamma_N=s$ for some $s\in \CS_N$, where $s\le N^{1/4}$, then (recall the second line of \eqref{E:MIX}) $L^{(N)}_{T-}=\alpha+\frac{s}{N}$ and $\frac{s}{N}\ll 1$; thus $L^{(N)}_{T-}$ is close to $\alpha$.  Hence $\tau^\beta_N>T$ (at least if $N>(\beta-\alpha)^{-4/3}$)
and $\tau^\alpha_N$ should be close to $T$; it should only take a short amount
of time for $L^{(N)}$ to increase the extra distance (which is at most $s/N$) past $\alpha$.
\begin{lemma}\label{L:TTimes} We have that
\begin{equation*} \varlimsup_{N\nearrow \infty}\sup_{\substack{s\in \CS_N \\ s\le N^{1/4}}}\tilde \BE_N\left[T-\tau^\alpha_N\bigg|\gamma_N\right]\chi_{\{\gamma_N=s\}}=0. \end{equation*}
\end{lemma}
\noindent Let's rigorously put all of these thoughts together.  Assume
that $N> (\beta-\alpha)^{-4/3}$ and $0<\gamma_N\le N^{1/4}$.  Then
$0\le \tau^\alpha_N\le T\le \tau^\beta_N$ and $0\le L^{(N)}_{\tau^\alpha_N-}\le L^{(N)}_{T-}= \alpha+\tfrac{\gamma_N}{N}<\beta$.  Hence
\begin{equation*} \int_{s\in [\tau^\alpha_N,\tau^\beta_N]\cap [0,T)}e^{-\rate s}d\tL^{(N)}_s
= e^{-\rate T}\left(\tL^{(N)}_{T-}-\tL^{(N)}_{\tau^\alpha_N-}\right) + \int_{s\in [\tau^\alpha_N,T)}\left(e^{-\rate s}-e^{-\rate T}\right) d\tL^{(N)}_s. \end{equation*}
Note that
\begin{equation*} \tL^{(N)}_{T-}=\frac{L^{(N)}_{T-}-\alpha}{\beta-\alpha} = \frac{1}{\beta-\alpha}\frac{\gamma_N}{N} \qquad \text{and}\qquad \tL^{(N)}_{\tau^\alpha_N-}=\frac{\left(L^{(N)}_{\tau^\alpha_N-}-\alpha\right)^+}{\beta-\alpha}. \end{equation*}
Thus
\begin{equation*} \int_{s\in [\tau^\alpha_N,\tau^\beta_N]\cap [0,T)}e^{-\rate s}d\tL^{(N)}_s = \frac{e^{-\rate T}}{\beta-\alpha}\frac{\gamma_N}{N} + \err_N \end{equation*}
where
\begin{equation*} \err_N = -e^{\rate T}\frac{\left(L^{(N)}_{\tau^\alpha_N-}-\alpha\right)^+}{\beta-\alpha} + \int_{s\in [\tau^\alpha_N,T)}e^{-\rate s}\lb 1- e^{-\rate (T-s)}\rb d\tL^{(N)}_s. \end{equation*}
If $\tau^\alpha_N>0$, then $L^{(N)}_{\tau^\alpha_N-}\le \alpha$.  Thus
\begin{equation*} \left(L^{(N)}_{\tau^\alpha_N-}-\alpha\right)^+ \le \frac{\gamma_N}{N}\chi_{\{\tau^\alpha_N=0\}}
=\frac{\gamma_N}{N}\chi_{\{T-\tau^\alpha_N=T\}}
\le \frac{1}{T}\frac{\gamma_N}{N}\left(T-\tau^\alpha_N\right). \end{equation*}
Similarly,
\begin{multline*} 0\le \int_{s\in [\tau^\alpha_N,T)}e^{-\rate s}\lb 1- e^{-\rate (T-s)}\rb d\tL^{(N)}_s \le \rate(T-\tau^\alpha_N)\left(\tL^{(N)}_{T-}-\tL^{(N)}_{\tau^\alpha_N-}\right)
\le \frac{\rate}{\beta-\alpha}(T-\tau^\alpha_N)\left(L^{(N)}_{T-}-\alpha\right)\\
= \frac{\rate}{\beta-\alpha}(T-\tau^\alpha_N)\frac{\gamma_N}{N} \end{multline*}
(we use here the fact that $\tL_{\tau^\alpha_N-}\ge 0$ and that $e^{-x}\ge 1-x$ for all $x\ge 0$).
Combining things, we get that
\begin{equation*} |\err_N|\le \frac{1}{\beta-\alpha}\lb \frac{1}{T}+\rate\rb(T-\tau^\alpha_N)\frac{\gamma_N}{N} \end{equation*}
on $\lb 0<\gamma_N\le N^{1/4}\rb$ if $N>(\beta-\alpha)^{-4/3}$.
We then have
\begin{proof}[Proof of Lemma \ref{L:hasymp}] For $s\in \CS_N$ such that $s\le N^{1/4}$, we have that
\begin{equation*} \Err_1(s,N)= (\beta-\alpha)e^{\rate T}\frac{\tilde \BE_N\left[\err_N\big|\gamma_N\right]}{\frac{\gamma_N}{N}}\chi_{\{\gamma_N=s\}} \le e^{\rate T}\lb \frac{1}{T} + \rate\rb \BE_N[T-\tau^\alpha_N|\gamma_N]\chi_{\{\gamma_N=s\}} \end{equation*}
if $N>(\beta-\alpha)^{-4/3}$.
Combine \eqref{E:AA}, the preceding calculations, and Lemma \ref{L:TTimes}. \end{proof}

We now need to prove Lemma \ref{L:TTimes}.  This is a moderately complex step.
The first problem is that by conditioning on $\gamma_N$, we are conditioning
on the value of $L^{(N)}$ near the \emph{endpoint} of the interval $[0,T)$
of interest.  The second problem is that we have a large amount
of randomness; $L^{(N)}$ can be decomposed into $N$ (independent) processes,
one corresponding to each name.

We shall resolve these issues by using the martingale problem to decompose
$L^{(N)}$ into a (reverse-time) zero-mean martingale and a term of bounded
variation\footnote{Much of our notation will thus be in reverse time.}.  We will use a martingale inequality to show that the martingale
part is small.  Thus the behavior of $L^{(N)}$ near $T$ will be given by
the bounded-variation part, which we can analyze via straightforward calculations.

Define now
\begin{equation*} \Zn_t \Def \chi_{\{\tau_n< T-t\}}=\chi_{(t,\infty]}(T-\tau_n) \qquad t\in [0,T) \end{equation*}
for each positive integer $n$ (note that the $\Zn$'s are right-continuous).
Also define $\gilt_t \Def \sigma\{\Zn_s: 0\le s\le t,\, n\in \{1,2\dots\}\}$
for all $t\in [0,T)$.  Observe that
\begin{equation*} L^{(N)}_{t-} = \frac{1}{N}\sum_{n=1}^N \chi_{[0,t)}(\tau_n)=\frac{1}{N}\sum_{n=1}^N \Zn_{T-t}. \end{equation*}
for all $t\in (0,T]$.

Let's now localize in time.  Let $T^*\in (0,T)$ be such that $F((T-T^*)-)>0$;
Assumption \ref{A:density} ensures that this is possible.
For all $t\in [0,T^*]$, define
\begin{align*} 
A^{(n)}_t &= -\int_{r\in [T-t,T)} \frac{1}{F(r)} \Zn_{(T-r)-}dF_r\\
M^{(n)}_t&\Def \Zn_t-\chi_{\{\tau_n<T\}}-A^{(n)}_t \end{align*}
(essentially, $A^{(N)}$ is the integral of the hazard function).
For future reference, we calculate that for any $t\in [0,T)$,
\begin{equation*} \Zn_{(T-t)-} = \lim_{s\searrow t}\Zn_{T-s} = \lim_{s\searrow t}\chi_{\{\tau_n<s\}} = \lim_{s\searrow t}\chi_{(0,s)}(\tau_n) = \chi_{(0,t]}(\tau_n). \end{equation*}
Note that by definition of $T^*$, 
\begin{equation}\label{E:Abound} \left|\frac{1}{F(r-)}\Zn_{(T-r)-}\right|\le \frac{1}{F((T-T^*)-)}<\infty \end{equation}
for all $r\in [T-T^*,T)$; thus $A^{(n)}$ is well-defined, finite, right-continuous, and it has left-hand limits.
\begin{lemma}\label{L:wmx} For every $n\in \{1,2\dots N\}$, $M^{(n)}$ is a $\tilde \BP_N$-zero-mean-martingale with
respect to $\{\gilt_t;\, t\in [0,T^*]\}$; i.e., for $0\le s\le t\le T^*$,
$\tilde \BE_N[M^{(n)}_t|\gilt_s]=M^{(n)}_s$.
\end{lemma}
\begin{proof} Fix $n$ as specified.   Clearly $M^{(n)}$ is adapted to $\{\gilt_t;\, t\in [0,T^*]\}$.
By \eqref{E:Abound}, we have that $A^{(n)}$ is also bounded, so $M^{(n)}_t$ is $\tilde \BP_N$-integrable for each $t\in [0,T^*]$.

We next compute some transition probabilities.  Fix $s$ and $t$ in $[0,T^*]$ such that $s\le t$.  Then $(t,\infty]\subset (s,\infty]$, so $\Zn_t\le \Zn_s$; hence $\Zn$ is nonincreasing.  This implies that
\begin{equation} \label{E:Zmon} \{\Zn_s=0\}\subset \{\Zn_t=0\} \qquad \text{and}\qquad  \{\Zn_t=1\}\subset \{\Zn_s=1\}. \end{equation}

Fix $0\le s_1<s_2\dots s_n\le s$ and $\{z_n\}_{n=1}^n\subset \{0,1\}$.  From
\eqref{E:Zmon}, we immediately have that
\begin{align*} &\tilde \BP_N\lb \Zn_t=0,\, \Zn_s=0,\, \Zn_{s_1}=z_1,\, \Zn_{s_2}=z_2\dots \Zn_{s_n}=z_n\rb \\
&\qquad =\tilde \BP_N\lb \Zn_s=0,\, \Zn_{s_1}=z_1,\, \Zn_{s_2}=z_2\dots \Zn_{s_n}=z_n\rb \\
&\tilde \BP_N\lb \Zn_t=1,\, \Zn_s=0,\, \Zn_{s_1}=z_1,\, \Zn_{s_2}=z_2\dots \Zn_{s_n}=z_n\rb =0. \end{align*}
A similar computation which also uses the definition of $\Zn$ gives us that
\begin{multline*} \tilde \BP_N\lb \Zn_t=1,\, \Zn_s=1,\, \Zn_{s_1}=z_1,\, \Zn_{s_2}=z_2\dots \Zn_{s_n}=z_n\rb \\
= \tilde \BP_N\lb \Zn_t=1,\, \Zn_{s_1}=z_1,\, \Zn_{s_2}=z_2\dots \Zn_{s_n}=z_n\rb
=\tilde \BP_N\{\tau_n< T-t\} \prod_{k=1}^n \delta_1(\{z_k\})\end{multline*}
A final computation (again using the definition of $\Zn$) gives us that
\begin{equation*} \tilde \BP_N\lb \Zn_t=0,\, \Zn_s=1,\, \Zn_{s_1}=z_1,\, \Zn_{s_2}=z_2\dots \Zn_{s_n}=z_n\rb =\tilde \BP_N\{T-t\le \tau_n< T-s\} \prod_{k=1}^n \delta_1(\{z_k\})\end{equation*}

With some manipulations, and using the fact that the $\Zn$'s are $\tilde \BP_N$-independent, we get that
\begin{equation}\label{E:condprob}\begin{aligned} \tilde \BP_N\lb \Zn_t=0\big| \gilt_s\rb &= \chi_{\{0\}}(\Zn_s)+\frac{\tilde \BP_N\{T-t\le \tau_n< T-s\}}{\tilde \BP_N\{\tau_n<T-s\}}\chi_{\{1\}}(\Zn_s)\\
\tilde \BP_N\lb \Zn_t=1\big| \gilt_s\rb &= \frac{\tilde \BP_N\{\tau_n< T-t\}}{\tilde \BP_N\{\tau_n< T-s\}}\chi_{\{1\}}(\Zn_s). \end{aligned}\end{equation}
Since $s<T^*$, 
\begin{equation*} \tilde \BP_N\{\tau_n< T-s\} = \frac{\alpha}{F(T)}\BP\{ \tau_n<T-s\}  = \frac{\alpha}{F(T)}F((T-s)-)\ge \frac{\alpha}{F(T)}F((T-T^*)-)>0; \end{equation*}
thus the expressions on the right of \eqref{E:condprob} are well-defined.
Proceeding, we compute that
\begin{equation*} \tilde \BE_N[\Zn_t|\gilt_s] = \frac{\tilde \BP_N\{\tau_n< T-t\}}{\tilde \BP_N\{\tau_n< T-s\}}\Zn_s \end{equation*}
and hence\footnote{Under normalization, $\mu$ and $\tilde \mu^*_\alpha$ agree
on $\Borel[0,T)$.}
\begin{equation}\label{E:LZ} \tilde \BE_N[\Zn_t|\gilt_s]-\Zn_s = \frac{\tilde \BP_N\{\tau_n< T-t\}-\tilde \BP_N\{\tau_n< T-s\}}{\tilde \BP_N\{\tau_n< T-s\}}\Zn_s =\frac{F((T-t)-)-F((T-s)-)}{F((T-s)-)}\Zn_s.\end{equation}

Again fix $s$ and $t$ in $[0,T^*]$ such that $s\le t$. For each positive integer $m$, define $r^m_k \Def s+(k/m)(t-s)$ for $k\in \{0,1\dots m\}$.  Using \eqref{E:LZ}, we can write that $\Zn_t-\Zn_s = \CA_m+\CM_m$ where
\begin{equation*} \CA_m = \sum_{k=0}^{m-1}\lb \tilde \BE_n\left[\Zn_{r^m_{k+1}}\big|\gilt_{r^m_k}\right]-\Zn_{r^m_k}\rb \qquad \text{and}\qquad
\CM_m = \sum_{k=0}^{m-1}\lb \Zn_{r^m_{k+1}}-\tilde \BE_n\left[\Zn_{r^m_{k+1}}\big|\gilt_{r^m_k}\right]\rb. \end{equation*}
For $C\in \gilt_s$,
\begin{equation}\label{E:miss} \tilde \BE_N\left[\lb \Zn_t-\Zn_s-\CA_m\rb \chi_C\right]
=\tilde \BE_N\left[\CM_m\chi_C\right]=0. \end{equation}

We now need to show that $\tilde \BP_N$-a.s.,
\begin{equation}\label{E:ggoal} \lim_{m\nearrow \infty}\CA_m = -\{A^{(n)}_t-A^{(n)}_s\} \end{equation}
This will require a bit of care.  We first rewrite $\CA_m$ as a integral;
\begin{equation*} \CA_m = -\sum_{k=0}^{m-1}\int_{r\in [T-r^m_{k+1},T-r^m_k)}\frac{1}{F((T-r^m_k)-)}\Zn_{r^m_k}dF_r = \int_{r\in [T-t,T-s)}\phi^m(r,\tau_n)dF_r \end{equation*}
where
\begin{equation*} \phi^m(r,t') \Def \sum_{k=0}^{m-1}\chi_{[T-r^m_{k+1},T-r^m_k)}(r)
\frac{1}{F((T-r^m_k)-)}\chi_{\{t'< T-r^m_k\}} \end{equation*}
for all $r\in [T-t,T-s)$ and $t'\in I$.
Defining
\begin{equation*} \tilde \phi(r,t') \Def \frac{1}{F(r-)}\chi_{(0,r)}(t') \end{equation*}
for all $r\in [T-t,T-s]$ and $t'\in I$, we thus have that
\begin{equation*} \phi^m(r,t') =\sum_{k=0}^{m-1}\chi_{[T-r^m_{k+1},T-r^m_k)}(r)\tilde \phi(T-r^m_k,t') \end{equation*}
for all $r\in [T-t,T-s)$ and $t'\in I$.  For $r\in [T-t,T-s)$, $F(r-)\ge F((T-t)-)\ge F((T-T^*)-)>0$; thus $\tilde \phi$ and the $\phi^m$'s are all uniformly bounded.  It is fairly easy to see that 
\begin{equation*} \lim_{m\to \infty}\phi^m(r,t') = \frac{1}{F(r)}\chi_{(0,r]}(t') \end{equation*}
for all $r\in [T-t,T-s)$ and all $t'\in I$.  Thus by dominated convergence,
\begin{equation*} \lim_{m\to \infty} \CA_m = -\int_{r\in [T-t,T-s)}\frac{1}{F(r)}\chi_{(0,r]}(\tau_n)dF_r, \end{equation*}
and \eqref{E:ggoal} follows.

Taking the limit in \eqref{E:miss}, we now have that
\begin{equation*} \tilde \BE_N\left[\lb (\Zn_t-A_t(\tau_n))-(\Zn_s-A_s(\tau_n))\rb \chi_C\right]=0, \end{equation*}
which is the martingale property.
Finally, since $M^{(n)}_0=0$, we have that $M^{(n)}$ is zero-mean.
\end{proof}

Let's now recombine things.  Set
\begin{equation*} \tilde M^{(N)}_t \Def \frac{1}{N}\sum_{n=1}^N M^{(n)}_t \qquad \text{and}\qquad \tilde A^{(N)}_t \Def \frac{1}{N}\sum_{n=1}^N A_t(\tau_n) \end{equation*}
for $t\in [0,T)$.  Note that
\begin{equation}\label{E:gmeas} L^{(N)}_{T-} = \frac{1}{N}\sum_{n=1}^N \Zn_0. \end{equation}
We next rewrite $\tau^\alpha_N$ as a stopping time with respect to $\{\gilt_t;\, t\in [0,T)\}$.  Set
\begin{equation*} \vrho^\alpha_N \Def \inf\lb t\in [0,T): L^{(N)}_{(T-t)-}\le \frac{\grandn}{N}\rb\wedge T = \inf\lb t\in [0,T): \frac{1}{N}\sum_{n=1}^N \Zn_t\le \frac{\grandn}{N}\rb\wedge T; \end{equation*}
then\footnote{If $L^{(N)}_0>\alpha$, then $\vrho^\alpha_N=T$ and $\tau^\alpha_N=0$.
Assume next that $L^{(N)}_0\le \alpha$.  Since $L^{(N)}$ is piecewise-constant and right-continuous,
we must have that $\tau^\alpha_N>0$.  At time $\tau^\alpha_N$, we have that
$L^{(N)}_{\tau^\alpha_N}>\alpha$ and $L^{(N)}_{\tau^\alpha_N-}\le \alpha$;
see Figure \ref{fig:typ}.  Since $L^{(N)}$ takes values only in $\Z/N$, we
have that $L^{(N)}_{\tau^\alpha_N-}\le \tfrac{\lfloor N\alpha\rfloor}{N}$
and $L^{(N)}_{\tau^\alpha_N}\ge \tfrac{\lfloor N\alpha\rfloor+1}{N}$.
Thus $\rho^\alpha_N=T-\tau^\alpha_N$, as claimed.} $\vrho^\alpha_N=T-\tau^\alpha_N$.  Furthermore, $\vrho^\alpha_N$ is a $\{\gilt_t;\, t\in [0,T)\}$-stopping time.

It will help to truncate $\vrho^\alpha_N$ at $T^*$; set
$\tilde \vrho^\alpha_N \Def \vrho^\alpha_N\wedge T^*$; this is also a $\{\gilt_t;\, t\in [0,T)\}$-stopping time and $\tilde \vrho^\alpha_N\le T^*$.
Thus
\begin{equation*} L^{(N)}_{(T-\tilde \vrho^\alpha_N)-} = L^{(N)}_{T-} + \tilde A^{(N)}_{\tilde \vrho^\alpha_N} + \tilde M^{(N)}_{\tilde \vrho^\alpha_N}. \end{equation*}
If $\gamma_N>0$, then $L^{(N)}_{T-}>\alpha$, and since $\tilde \vrho^\alpha_N\le \vrho^\alpha_N$, we have that $L^{(N)}_{(T-\tilde \vrho^\alpha_N)-}\ge \frac{\grandn}{N}$ and consequently
\begin{multline*} - \tilde A^{(N)}_{\tilde \vrho^\alpha_N} = L^{(N)}_{T-} - L^{(N)}_{(T-\tilde \vrho^\alpha_N)-} + \tilde M^{(N)}_{\tilde \vrho^\alpha_N}
\le L^{(N)}_{T-} - \frac{\grandn}{N} + \tilde M^{(N)}_{\tilde \vrho^\alpha_N}
\le L^{(N)}_{T-} - \alpha + \frac{1}{N}+ |\tilde M^{(N)}_{\tilde \vrho^\alpha_N}|\\
\le \gamma_N + \frac{1}{N}+ |\tilde M^{(N)}_{\tilde \vrho^\alpha_N}|. \end{multline*}

Let's now use the fact that $\frac{1}{N}\sum_{n=1}^N \Zn_{\tilde \vrho^\alpha_N}\ge \frac{\grandn}{N}$ to bound $\tilde A^{(N)}_{\tilde \vrho^\alpha_N}$.  As we pointed out in the proof of Lemma \ref{L:wmx}, the $\Zn$'s are nonincreasing.  Also, $F\le 1$.  Thus for $N\ge 2/\alpha$ (which implies that $\grandn/N\ge \alpha/2$),
we have the following string of inequalities.
\begin{multline*} -\tilde A^{(N)}_{\tilde \vrho^\alpha_N} 
=\int_{r\in [T-t,T)}\frac{1}{F(r-)}\left(\frac{1}{N}\sum_{n=1}^N \Zn_{(T-r)-}\right)dF_r 
\ge \left(\frac{1}{N}\sum_{n=1}^N \Zn_{\vrho^\alpha_N-}\right)\int_{r\in [T-\tilde \vrho^\alpha_N,T)}dF_r\\
\ge \left(\frac{1}{N}\sum_{n=1}^N \Zn_{\vrho^\alpha_N}\right)\lb F(T-)-F((T-\tilde \vrho^\alpha_N)-)\rb \ge \frac{\alpha}{2}\ff(\tilde \rho^\alpha_N) \end{multline*}
where we have defined
\begin{equation*} \ff(\eps) \Def F(T-)-F((T-\eps)-) \end{equation*}
for all $\eps\in (0,T)$.  Note that $\lim_{\eps \searrow 0}\ff(\eps)=0$, and, thanks to Assumption \ref{A:density}, $\ff(\eps)>0$ for $\eps\in (0,T)$.
Thus
\begin{equation} \label{E:CW} \ff\left(\tilde \vrho^\alpha_N\right)\chi_{\{\gamma_N>0\}}\le \frac{2}{\alpha}\lb \gamma_N + \frac{1}{N} + \left|\tilde M^{(N)}_{\tilde \vrho^\alpha_N}\right|\rb\chi_{\{\gamma_N>0\}}
\le \frac{2}{\alpha}\lb \gamma_N^+ + \frac{1}{N} + \left|\tilde M^{(N)}_{\tilde \vrho^\alpha_N}\right|\rb\end{equation}
if $N>2/\alpha$.
\begin{proof}[Proof of Lemma \ref{L:TTimes}]  We begin by taking conditional expectations of \eqref{E:CW}.  Note that $\gamma_N$ is $\gilt_0$-measurable (see \eqref{E:gmeas}).
We have
\begin{equation*} \tilde \BE_N\left[\ff\left(\tilde \vrho^\alpha_N\right)\big|\gilt_0\right]\chi_{\{\gamma_N>0\}} \le \frac{2}{\alpha}\lb \gamma_N^+ + \frac{1}{N} + \tilde \BE_N\left[\left|\tilde M^{(N)}_{\tilde \vrho^\alpha_N}\right|\bigg| \gilt_0\right]\rb . \end{equation*}
By Jensen's inequality,
\begin{equation*} \tilde \BE_N\left[\left|\tilde M^{(N)}_{\tilde \vrho^\alpha_N}\right|\bigg| \gilt_0\right]\le \tilde \BE_N\left[\left(\tilde M^{(N)}_{\tilde \vrho^\alpha_N}\right)^2\bigg| \gilt_0\right]^{1/2} \end{equation*}
$\tilde \BP_N$-a.s.
We can now use optional sampling;
\begin{multline*} \tilde \BE_N\left[\left(\tilde M^{(N)}_{\tilde \vrho^\alpha_N}\right)^2\bigg| \gilt_0\right] \le \tilde \BE_N\left[\left(\tilde M^{(N)}_{T_2^*}\right)^2\bigg| \gilt_0\right]
= \frac{1}{N^2}\sum_{n=1}^N \tilde \BE_N\left[\left(M^{(n)}_{T_2^*}\right)^2\bigg| \gilt_0\right]\\
\le \frac{3}{N^2}\sum_{n=1}^N \lb \tilde \BE_N\left[\left(\Zn_{T_2^*}\right)^2\bigg| \gilt_0\right]+\tilde \BE_N\left[\chi^2_{\{\tau_n<T\}}\bigg| \gilt_0\right]
+\tilde \BE_N\left[\left(A^{(n)}_{T^*}\right)^2\bigg| \gilt_0\right]\rb\\
\le \frac{3}{N}\lb 2 + \frac{1}{F^2((T-T^*)-)}\rb  \end{multline*}
$\tilde \BP_N$-a.s.
We have used here the fact that the $M^{(n)}$'s are independent, the explicit
formula for $M^{(n)}$, and \eqref{E:Abound}.
Summarizing thus far, we have that
\begin{equation*}  \tilde \BE_N\left[\ff\left(\tilde \vrho^\alpha_N\right)\big|\gilt_0\right]\chi_{\{\gamma_N>0\}}\le \frac{2}{\alpha}\lb \left(L^{(N)}_{T-}-\alpha\right)^+ + \frac{1}{N}  + \sqrt{\frac{3}{N}\lb 2 + \frac{1}{F^2((T-T^*)-)}\rb }\rb \end{equation*}
$\tilde \BP_N$-a.s.  As we pointed out earlier, $\sigma\{\gamma_N\} = \sigma\{L^{(N)}_{T-}\} \subset \gilt_0$, so by iterated conditioning, we next have that
\begin{equation*}  \tilde \BE_N\left[\ff\left(\tilde \vrho^\alpha_N\right)\big|\gamma_N\right]\chi_{\{\gamma_N>0\}}\le \frac{2}{\alpha}\lb \gamma_N^+ + \frac{1}{N}  + \sqrt{\frac{3}{N}\lb 2 + \frac{1}{F^2((T-T^*)-)}\rb}\rb \end{equation*}

Fix now $\eps\in (0,T^*)$.  If $\tilde \vrho^\alpha_N<\eps$, then in fact $T-\tau^\alpha_N = \tilde \vrho^\alpha_N<\eps$.  On the other hand, if $\tilde \vrho^\alpha_N\ge \eps$, then 
\begin{equation*} \gamma_N = L^{(N)}_{T-}-\alpha \ge \frac{\grandn + 1}{N}-\alpha >0; \end{equation*}
thus
\begin{equation*} \chi_{\{\tilde \vrho^\alpha_N>\eps\}} \le \frac{1}{\ff(\eps)}\ff(\tilde \vrho^\alpha_N)\chi_{\{\gamma_N>0\}}. \end{equation*}
Hence
\begin{multline*} \tilde \BE_N\left[T-\tau^\alpha_N\bigg|\gamma_N\right]
\le \eps + T \frac{\tilde \BE_N\left[\ff\left(\tilde \vrho^\alpha_N\right)\chi_{\{\gamma_N>0\}}\bigg| L^{(N)}_{T-}\right]}{\ff(\eps)} \\
\le \eps + \frac{2T}{\alpha\ff(\eps)}\lb \gamma_N^+ + \frac{1}{N} + \sqrt{\frac{3}{N}\lb 2 + \frac{1}{F^2((T-T^*)-)}\rb}\rb \end{multline*}
$\tilde \BP_N$-a.s. 
In other words,
\begin{equation*} \sup_{\substack{s\in \CS_N \\ s\le N^{1/4}}}\tilde \BE_N\left[T-\tau^\alpha_N\big| \gamma_N\right]\chi_{\{\gamma_N=s\}}\le 
\eps + \frac{2T}{\alpha\ff(\eps)}\lb \frac{1}{N^{3/4}} + \frac{1}{N} + \sqrt{\frac{3}{N}\lb 2 + \frac{1}{F^2((T-T^*)-)}\rb }\rb. \end{equation*}
Let $N\nearrow \infty$ and then let $\eps\searrow 0$.
\end{proof}

\section{Proofs}\label{S:Proofs}

We here give the deferred proofs.

\begin{proof}[Proof of Lemma \ref{L:probasymp}]  To begin, recall Stirling's formula.  Let $\tilde \Err_1:(-1,\infty)\to \R$ be defined by
\begin{equation*} \Gamma(x+1) \Def \int_{u=0}^\infty u^x e^{-u}du = \left(\frac{x}{e}\right)^x\sqrt{2\pi x}\lb 1+\tilde \Err_1(x)\rb; \end{equation*}
for all $x>-1$; then $\lim_{x\to \infty}|\tilde \Err_1(x)|=0$.
Then for any $s=n-N\alpha\in \CS_N$,
\begin{multline*} \tilde \BP_N\lb \gamma_N=s\rb
= \tilde \BP_N\lb \text{$n$ of the $\tau$'s are in $[0,T)$ and $N-n$ are in $[T,\infty)$}\rb \\
= \binom{N}{n}\alpha^n (1-\alpha)^{N-n}
= \frac{\Gamma(N+1)}{\Gamma(N\alpha+s+1)\Gamma(N(1-\alpha)-s+1)}\alpha^{N\alpha+s}(1-\alpha)^{N(1-\alpha)-s}
= A(N)B(s,N) \end{multline*}
where
\begin{align*} A(N) &=  \frac{\Gamma(N+1)}{\Gamma(N\alpha+1)\Gamma(N(1-\alpha)+1)}\alpha^{N\alpha}(1-\alpha)^{N(1-\alpha)}\\
B(s,N) &= \frac{\Gamma(N\alpha+1)\Gamma(N(1-\alpha)+1)}{\Gamma(N\alpha+s+1)\Gamma(N(1-\alpha)-s+1)}\alpha^s(1-\alpha)^{-s}. \end{align*}
Let's now use Stirling's formula.
We have
\begin{align*} A(N)&= \frac{\left(\frac{N}{e}\right)^N}{\left(\frac{N\alpha}{e}\right)^{N\alpha}\left(\frac{N(1-\alpha)}{e}\right)^{N(1-\alpha)}}\frac{\sqrt{2\pi N}}{\sqrt{2\pi N\alpha}\sqrt{2\pi N(1-\alpha)}}\alpha^{N\alpha}(1-\alpha)^{N(1-\alpha)} \\
&\qquad \times \frac{1+\tilde \Err_1(N)}{\{1+\tilde \Err_1(N\alpha)\}\{1+\tilde \Err_1(N(1-\alpha))\}} \\
&=\frac{1}{\sqrt{2\pi N\alpha(1-\alpha)}}\frac{1+\tilde \Err_1(N)}{\{1+\tilde \Err_1(N\alpha)\}\{1+\tilde \Err_1(N(1-\alpha))\}} \end{align*}
Thus
\begin{equation*} A(N)=\frac{1+\tilde \Err_2(N)}{\sqrt{2\pi N\alpha(1-\alpha)}}. \end{equation*}
where $\lim_{N\to \infty}\tilde \Err_2(N)=0$.
To find the asymptotics of $B$, we first let $\tilde \Err_3:(-1,\infty)$ be such that
\begin{equation*} \ln (1+x) = x+\tilde \Err_3(x). \qquad x>-1 \end{equation*}
Then there is a $\KK>0$ such that $|\tilde \Err_3(x)|\le \KK x^2$ for all $x\in (-1/2,1/2)$.
Again using Stirling's formula, we have that
\begin{align*} B(s,N) &= \frac{\left(\frac{N\alpha}{e}\right)^{N\alpha}}{\left(\frac{N\alpha+s}{e}\right)^{N\alpha+s}}\frac{\left(\frac{N(1-\alpha)}{e}\right)^{N(1-\alpha)}}{\left(\frac{N(1-\alpha)-s}{e}\right)^{N(1-\alpha)-s}}\sqrt{\left(\frac{N\alpha}{N\alpha+s}\right)\left(\frac{N(1-\alpha)}{N(1-\alpha)-s}\right)}\alpha^s (1-\alpha)^{-s}\\
&\qquad \times \frac{1+\tilde \Err_1(N\alpha)}{\lb 1+\tilde \Err_1(N\alpha+s)\rb\lb 1+\tilde \Err_1(N(1-\alpha)-s)\rb}\\
      &= \left(\frac{\alpha}{\alpha+\frac{s}{N}}\right)^{N\alpha+s}\left(\frac{1-\alpha}{1-\alpha-\frac{s}{N}}\right)^{N\alpha-s}\frac{1}{\sqrt{\left(1+\frac{s}{\alpha N}\right)\left(1-\frac{s}{N(1-\alpha)}\right)}}\\
&\qquad \times \frac{1+\tilde \Err_1(N\alpha)}{\lb 1+\tilde \Err_1(N\alpha+s)\rb\lb 1+\tilde \Err_1(N(1-\alpha)-s)\rb}\\
& = \frac{1}{\left(1+\frac{s}{N\alpha}\right)^{N\alpha+s}\left(1-\frac{s}{N(1-\alpha)}\right)^{N(1-\alpha)-s}} \frac{1}{\sqrt{\left(1+\frac{s}{\alpha N}\right)\left(1-\frac{s}{N(1-\alpha)}\right)}}\\
&\qquad \times \frac{1+\tilde \Err_1(N\alpha)}{\lb 1+\tilde \Err_1(N\alpha+s)\rb\lb 1+\tilde \Err_1(N(1-\alpha)-s)\rb}\\
&= \exp\left[-(N\alpha+s)\left(\frac{s}{N\alpha} + \tilde \Err_2\left(\frac{s}{N\alpha}\right)\right)- (N(1-\alpha)-s)\left(-\frac{s}{N(1-\alpha)} + \tilde \Err_2\left(\frac{-s}{N(1-\alpha)}\right)\right)\right]\\
&\qquad \times \frac{1}{\sqrt{\left(1+\frac{s}{\alpha N}\right)\left(1-\frac{s}{N(1-\alpha)}\right)}}\times \frac{1+\tilde \Err_1(N\alpha)}{\lb 1+\tilde \Err_1(N\alpha+s)\rb\lb 1+\tilde \Err_1(N(1-\alpha)-s)\rb} \\
&= \exp\left[-\tilde \Err_4(s,N)\right]\frac{1}{\sqrt{\left(1+\frac{s}{\alpha N}\right)\left(1-\frac{s}{N(1-\alpha)}\right)}}\\
&\qquad \times \frac{1+\tilde \Err_1(N\alpha)}{\lb 1+\tilde \Err_1(N\alpha+s)\rb\lb 1+\tilde \Err_1(N(1-\alpha)-s)\rb} \end{align*}
where
\begin{equation*} \tilde \Err_4(s,N) \Def \frac{s^2}{N}\left(\frac{1}{\alpha}+\frac{1}{1-\alpha}\right) + (N\alpha+s)\tilde \Err_3\left(\frac{s}{N\alpha}\right)+(N(1-\alpha)-s)\tilde \Err_3\left(-\frac{s}{N(1-\alpha)}\right). \end{equation*}

Let's now combine things together.  We have that
\begin{equation*} \tilde \BP_N\{\gamma_N=s\}-\frac{1}{\sqrt{2\pi N\alpha(1-\alpha)}} = \frac{\tilde \Err_5(s,N)-1}{\sqrt{2\pi N\alpha(1-\alpha)}} \end{equation*}
where
\begin{align*} \tilde \Err_5(s,N)&=\frac{1+\tilde \Err_1(N)}{\{1+\tilde \Err_1(N\alpha)\}\{1+\tilde \Err_1(N(1-\alpha))\}}\exp\left[-\tilde \Err_4(s,N)\right]\frac{1}{\sqrt{\left(1+\frac{s}{\alpha N}\right)\left(1-\frac{s}{N(1-\alpha)}\right)}}\\
&\qquad \times\frac{1+\tilde \Err_1(N\alpha)}{\lb 1+\tilde \Err_1(N\alpha+s)\rb\lb 1+\tilde \Err_1(N(1-\alpha)-s)\rb}. \end{align*}
Note that if $|s|\le N^{1/4}$, then
\begin{equation*} \left|\frac{s}{N\alpha}\right|\le \frac{1}{\alpha N^{3/4}} \qquad \text{and}\qquad \left|\frac{s}{N(1-\alpha)}\right|\le \frac{1}{(1-\alpha) N^{3/4}}. \end{equation*}
Thus if $|s|\le N^{1/4}$, then for $N$ large enough
\begin{equation*} |\tilde \Err_4(s,N)| \le \frac{K}{\sqrt{N}}. \end{equation*}
The claimed statement follows.
\end{proof}

Let's now start to prove Proposition \ref{P:explicit}.
First, define
\begin{equation*} \fI_\circ(\alpha') \Def \inf\lb H(\tilde \mu|\mu): \tilde \mu\in \Pspace(I), \tilde \mu[0,T)=\alpha'\rb \end{equation*}
for all $\alpha'\in [0,1]$.
\begin{lemma}\label{L:hear} We have that
\begin{equation*} \fI_\circ(\alpha') = \hbar(\alpha',F(T-))=H(\tilde \mu^*_{\alpha'}|\mu), \end{equation*}
where $\tilde \mu_{\alpha'}^*$ is given by \eqref{E:dacc}.
\end{lemma}
\begin{proof}
Fix $\mu'\in \Pspace(I)$ such that $\mu'[0,T)=\alpha'$.
If $\mu'$ is not absolutely continuous
with respect to $\mu$, then $H(\mu'|\mu)=\infty$; thus we assume
that $\mu'$ is absolutely continuous with respect to $\mu$.
Define
\begin{equation*} f(x) \Def \begin{cases} x \ln x &\text{for $x>0$} \\
0 &\text{if $x=0$.} \end{cases}\end{equation*}
Then $f$ is convex on $[0,\infty)$.  Recall that Assumptions \ref{A:density} and \ref{A:IG} imply that $F(T)\in (0,1)$.  Thus we can write (using Jensen's inequality) that
\begin{multline*} H(\mu'|\mu) = \int_{t\in I}f\left(\frac{d\mu'}{d\mu}(t)\right)\mu(dt)\\
=\mu[0,T)\int_{t\in [0,T)}f\left(\frac{d\mu'}{d\mu}(t)\right)\frac{\mu(dt)}{\mu[0,T)} + \mu[T,\infty]\int_{t\in [T,\infty]}f\left(\frac{d\mu'}{d\mu}(t)\right)\frac{\mu(dt)}{\mu[T,\infty]} \\
\ge \mu[0,T)f\left(\int_{t\in [0,T)}\frac{d\mu'}{d\mu}(t)\frac{\mu(dt)}{\mu[0,T)}\right) + \mu[T,\infty]f\left(\int_{t\in [T,\infty]}\frac{d\mu'}{d\mu}(t)\frac{\mu(dt)}{\mu[T,\infty]}\right) \\
= \mu[0,T)f\left(\frac{\mu'[0,T)}{\mu[0,T)}\right) + \mu[T,\infty]f\left(\frac{\mu'[T,\infty]}{\mu[T,\infty]}\right) = \hbar(\mu'[0,T),\mu[0,T)). \end{multline*}
We have equality here if and only if $\mu$-a.s.
\begin{equation*} \frac{d\mu'}{d\mu}(t) = C_1\chi_{[0,T)}(t) + C_2\chi_{[T,\infty]}(t), \end{equation*}
which holds if and only if $\mu' = \tilde \mu^*_{\alpha'}$.  Collecting things together,
we have the claimed result.\end{proof}

\begin{proof}[Proof of Proposition \ref{P:explicit}]
In light of Lemma \ref{L:hear}, we need to show that
\begin{equation} \label{E:minc} \inf_{\alpha'\ge \alpha}\fI_\circ(\alpha')=\fI(\alpha). \end{equation}
To do this, we observe from that $\fI_\circ$ is differentiable and that
\begin{equation*} \fI_\circ'(\alpha') = \ln \left(\frac{\alpha'}{1-\alpha'}\frac{1-F(T-)}{F(T-)}\right)=\ln \left(1+\frac{\alpha'-F(T-)}{F(T-)(1-\alpha')}\right) \end{equation*}
for all $\alpha'\in (0,1)$.   Thus $\fI'_\circ(\alpha')>0$ if $\alpha'>F(T-)$.  Assumption \ref{A:IG} then gives us \eqref{E:minc}.  Combining our arguments, we get the claimed result.
\end{proof}

\section{Appendix: Measurability}
We here verify that $\Prot_N$ is measurable.
Let $D_+$ be the collection of nondecreasing functions $\phi:\R\to [0,1]$ which are right-continuous and have left-hand limits and for which $\phi(t)=0$ for $t<0$.
For $\mu'\in \Pspace(I)$ and $t\ge 0$, define $\iota(\mu')(t) \Def \mu'[0,t]$ for $t\ge 0$ and $\iota(\mu')(t)=0$ for $t<0$.  In particular, $L^{(N)}=\iota(\nu^{(N)})$.  It is clear that $\iota:\Pspace(I)\to D_+$ and is a bijection (note that $\mu'\{\infty\} = 1- \lim_{t\nearrow \infty}\iota(\mu')(t)$; this
allows us to recover $\mu'\{\infty\}$ when writing down the inverse of $\iota$).  We can then topologize $D_+$ by pushing the topology of $\Pspace(I)$ forward through $\iota$; thus $\iota$ is continuous.   We also note that $\{\ell_n\}_{n=1}^\infty \subset D_+$ converges to $\ell\in D_+$ if and only if $\lim_{n\to \infty}\ell_n(t)=\ell(t)$
for all points $t\in [0,\infty)$ at which $\ell$ is continuous.  Thus $L^{(N)}$ is a $D_+$-valued random variable.  We next define
$\Phi^\circ_1:D_+\to D_+$ as 
\begin{equation*} \Phi^\circ_1(\phi)(t) \Def \frac{(\phi(t)-\alpha)^+-(\phi(t)-\beta)^+}{\beta-\alpha} \qquad t\in \R \end{equation*}
for all $\phi\in D_+$.  Thus $\tL^{(N)}=\Phi^\circ_1(L^{(N)})$.  By the above characterization of convergence in $D_+$,
we see that $\Phi^\circ_1$ is continuous; thus $\tL^{(N)}$ is also a $D_+$-valued
random variable.  Finally, define $\Phi^\circ_2 \Def D_+\to \R$ as 
\begin{equation*} \Phi^\circ_2(\ell) \Def \int_{s\in [0,T)} e^{-\rate s}d\ell(s) = e^{-\rate T}\ell(T-)+ \rate \int_{s\in (0,T)} e^{-\rate s}\ell(s)ds \end{equation*}
for all $\ell\in D_+$ (we define the $d\ell$ integral as a Lebesgue-Stieltjes
integral).  Then $\Prot_N = \Phi^\circ_2(\tL^{(N)})$.  We claim that
$\Prot_N$ is measurable (from $D_+$ to $\R$).
Let $\zeta\in C^\infty(\R;[0,1])$ be such that $\zeta(t)=1$
for $t\le -1$ and $\zeta(t)=0$ for $t\ge 0$.  For each positive integer $n$
and each $\ell\in D_+$, set
\begin{equation*} \tilde \Phi^n_2(\ell) \Def \int_{s\in \R} e^{-\rate s}\zeta(n(s-T))d\ell(s) = \int_{s=0}^\infty e^{-\rate s}\lb \rate \zeta(n(s-T))- n\dot \zeta(n(s-T))\rb \ell(s)ds. \end{equation*}
Clearly $\tilde \Phi^n_2:D_+\to \R$ is continuous.  Furthermore, by dominated convergence, $\lim_{n\to \infty}\tilde \Phi^n_2(\ell) = \Phi^\circ_2(\ell)$ for each $\ell$ (i.e., pointwise on $D_+$).  Being the pointwise limit of continuous functions, $\Phi^\circ_2$ is thus measurable.

Combining all of these arguments, we conclude that $\Prot_N$ is indeed a $\R$-valued random variable.  Clearly $\tL^{(N)}_T\le 1$, so $0\le\Prot_N\le 1$.

\bibliographystyle{alpha}
\def\cprime{$'$} \def\cprime{$'$} \def\cprime{$'$} \def\cprime{$'$}
  \def\polhk#1{\setbox0=\hbox{#1}{\ooalign{\hidewidth
  \lower1.5ex\hbox{`}\hidewidth\crcr\unhbox0}}}

\end{document}